\documentclass[letterpaper,10pt,twocolumn]{IEEEtran}
\usepackage{cite}
\usepackage{amsmath,amssymb,amsfonts,mathptmx}
\usepackage{amsthm}
\usepackage{graphicx}
\usepackage{enumerate}
\usepackage{epsfig}
\usepackage[font=small,labelfont=bf]{caption}
\usepackage{bm}
\usepackage{epstopdf}
\usepackage{algorithm,array}
\usepackage[noend]{algpseudocode}
\usepackage{multirow,multicol}

\usepackage{braket,enumitem,adjustbox,booktabs,lipsum}
\usepackage{color,slashbox}
\usepackage{float}
\usepackage{boondox-cal}
\newfloat{algorithm}{t}{lop}
\usepackage{pgf,tikz}
\usetikzlibrary{arrows}
\ifCLASSOPTIONcompsoc
\usepackage[caption=false, font=normalsize, labelfont=sf, textfont=sf]{subfig}
\else
\usepackage[caption=false, font=footnotesize]{subfig}
\fi

\newcommand{\oprocendsymbol}{\hbox{$\bullet$}}
\newcommand{\oprocend}{\relax\ifmmode\else\unskip\hfill\fi\oprocendsymbol}

\newtheorem{theorem}{Theorem}[section]
\newtheorem{proposition}[theorem]{Proposition}
\newtheorem{lemma}[theorem]{Lemma}
\newtheorem{corollary}[theorem]{Corollary}
\newtheorem{remark}[theorem]{Remark}

\renewcommand{\theenumi}{(\roman{enumi})}

\newcommand{\real}{\ensuremath{\mathbb{R}}}
\newcommand{\complex}{\ensuremath{\mathbb{C}}}

\newcommand{\tA}{\overline{A}}

\newcommand{\Hc}{\mathcal{H}}

\newcommand{\Wc}{\mathcal{W}}
\newcommand{\Xc}{\mathcal{X}}
\newcommand{\tXc}{\overline{\mathcal{X}}}
\newcommand{\Yc}{\mathcal{Y}}
\newcommand{\Zc}{\mathcal{Z}}
\newcommand{\Wci}{\mathcal{W}^{\infty}}
\newcommand{\Wcinv}{\mathcal{W}^{-1}}
\newcommand{\BB}{\vert B\vert \vert B\vert ^{\top}}
\newcommand{\T}{^{\top}}
\newcommand{\lo}{\tilde{\lambda}}
\newcommand{\wo}{\tilde{\omega}}
\newcommand{\fb}{\overline{f}}
\newcommand{\hb}{\overline{h}}

\newcommand{\until}[1]{\{1,\dots,#1\}}
\newcommand{\ontil}[1]{\{0,\dots,#1\}}

\newcommand\tr{\operatorname{tr}}
\newcommand\sub{\operatorname{sub}}
\newcommand\rank{\operatorname{rank}}
\renewcommand\det{\operatorname{det}}

\newcommand{\longthmtitle}[1]{\mbox{}{\textit{(#1).}}}

\renewcommand{\baselinestretch}{1}
\allowdisplaybreaks
\graphicspath{{epsfiles/}}

\title{\LARGE \bf Encoding Impact of Network Modification on
  Controllability via Edge Centrality Matrix\thanks{A preliminary version of this work
    appeared as~\cite{PVC-JC:20-acc} at the American Control
    Conference. This work was supported by ARO Award
    W911NF-18-1-0213.}}

\author{Prasad Vilas Chanekar$^{*}$ \quad Jorge
  Cort{\'e}s
  \thanks{$^{*}$Corresponding author.}  
  \thanks{The authors are with the Department of Aerospace and
    Mechanical Engineering, University of California, San Diego, La
    Jolla, CA 92093, USA, {\tt \{pchanekar, cortes\}@ucsd.edu}}%
}

\begin{document}
\maketitle

\begin{abstract}
  This paper develops tools to quantify the importance of agent
  interactions and its impact on global performance metrics for
  networks modeled as linear time-invariant systems.  We consider
  Gramian-based performance metrics and propose a novel notion of edge
  centrality that encodes the first-order variation in the metric with
  respect to the modification of the corresponding edge weight,
  including for those edges not present in the network.  The proposed
  edge centrality matrix (ECM) is additive over the set of inputs,
  i.e., it captures the specific contribution to each edge's
  centrality of the presence of any given actuator.  We provide a full
  characterization of the ECM structure for the class of directed
  stem-bud networks, showing that non-zero entries are only possible
  at specific sub/super-diagonals determined by the network size and
  the length of its bud.  We also provide bounds on the value of the
  trace, trace inverse, and log-det of the Gramian before and after
  single-edge modifications, and on the edge-modification weight to
  ensure the modified network retains stability.  Simulations show the
  utility of the proposed edge centrality notion and validate our
  results.
\end{abstract}

\section{Introduction}\label{introduction-1}

Network control systems find application in a wide range of domains
and activities, including social dynamics, energy systems, intelligent
transportation, and robotics.  In such scenarios, the network must
respond efficiently to the inputs of its authorized users while at the
same time remaining resilient against external interference or
malicious attacks.  As inputs are applied at nodes and propagated
through the interconnections, understanding the role of each node and
each edge on driving the network behavior is key. The complexity of
this problem is higher in the edge case than in the nodal case since
the number of edges scales quadratically with the number of nodes.  To
break down this complexity for agent interactions, this paper focuses
on providing energy-based edge centrality notions that, given a set of
inputs, allow us to quantify the relative impact of individual edges
on the network controllability properties.

\subsection*{Literature review}

Centrality notions aim to provide a way to quantify the relative
importance of nodes and edges in a complex network with respect to a
given performance metric, see c.f.~\cite{MOJ:10,MN:18,MB-CK:15}.  The
predominant focus on the role of nodes and the computational easiness
of node-based centrality measures makes these particularly popular in
the characterization of network properties.  Based on the topological
properties of the network, some commonly used nodal centrality
measures include degree~\cite{LCF:79,MN:18}, closeness~\cite{MN:18},
betweeness~\cite{LCF:77}, eigenvector~\cite{PB:87}, Katz~\cite{LK:53},
PageRank (Google)~\cite{SB-LP:12}, percolation~\cite{MP-MP-LH:13},
cross-clique~\cite{MGE-SPB:98}, Freeman~\cite{LCF:79},
topological~\cite{GR-ZZ:13}, Markov~\cite{SW-PS:03}, hub and
authority~\cite{JMK:99}, routing~\cite{SD-YE-RP:10},
subgraph~\cite{EE-JAR-05}, and total communicability~\cite{MB-CK:13}
centralities. In the case of edge centrality, notions include
betweeness centrality~\cite{UB:01-jms}, edge HITS centrality, and edge
total communicability centrality~\cite{FA-MB:15}.  These centrality
measures are based on topological considerations and connectivity
properties of the network, and in general overlook the role of the
dynamics of individual nodes and edges in driving network behavior. As
an example, one might argue that a densely connected node with a very
slow timescale for its dynamics might play a lesser role than a less
densely connected node with a faster dynamics.  

Dynamics-based centrality measures encompass fewer notions, mostly
limited to node
centrality~\cite{KF-NEL:16,GL-CA:19,EN-FP-JC:19-jcn,MS-SB-BB-NM:17},
and are based on performance
metrics~\cite{GY-JR-YL-CL-BL:12,FP-SZ-FB:14,THS-FLC-JL:16} based on
the spectral properties of the controllability Gramian~\cite{CTC:98}.
These energy-based metrics include the trace of the
Gramian~\cite{THS-FLC-JL:16,GL-CA:19}, the trace of its
inverse~\cite{THS-FLC-JL:16}, its
determinant~\cite{PCM-HIW:72,VT-MAR-GJP-AJ:16}, and its minimum
eigenvalue~\cite{FP-SZ-FB:14}. While controllability only captures the
ability to steer the network between any pair of states, such metrics
quantify the optimal energy required to do so, which allows for a more
nuanced accounting of the interplay between topology and dynamics in
determining centrality. In the cases of edges,~\cite{YG-MS-CS-NM:21}
proposes an edge centrality measure with respect to the $\Hc_2-$norm
for networks with continuous-time consensus dynamics having time
delays and structured uncertainties. Works~\cite{SZ-FP:17,SZ-FP:19}
characterize networks with diagonal controllability Gramian and also
propose pathways to design them for prescribed controllability
properties. Our previous work~\cite{PVC-EN-JC:21-tnse} proposes a
notion of Gramian-based edge centrality for directed topologies with
non-negative weights which encodes their role in energy transmission
throughout the network, irrespective of the input location.  Finally,
the work~\cite{GL-CA:20-ifac} studies conditions under which edge
modifications to networks with positive edge weights do not compromise
its stability and derives an upper bound on the allowable perturbation
weight. This work also studies the analytical characterization of
performance metrics such as coherence ($\Hc_2-$norm) and robustness
($\Hc_{\infty}-$norm) after edge perturbation.

\subsection*{Statement of contributions}

We consider networks described by linear discrete-time systems, where
the agent-to-agent connectivity is encoded by the system matrix. In
our treatment, the network adjacency matrix is not required to be
symmetric and the edge weights can have arbitrary sign.  Our first
contribution is the explicit computation of the first-order variation
of a generalized version of the Gramian matrix with respect to the
elements of the adjacency matrix. This result allows us to express the
gradients of various Gramian-based performance metrics in a unified,
computationally efficient, way, which in turn is the basis for the
introduction of the edge centrality matrix (ECM).  The notion of ECM
is tightly coupled with physically realizable energy-based system
properties and provides a measure of the relative importance of each
edge, including those not present in the network, on the performance
metrics. ECM is additive in the input space and therefore allows to
precisely identify the impact of individual network inputs in
determining the importance of each edge.  Our second contribution is
the characterization of the structure of ECM for the family of
directed stem-bud networks. Each of these networks is a combination of
a line network and ring network, and possesses a diagonal
controllability Gramian. We show that non-zero entries of the ECM are
only possible at specific sub/super- diagonals determined by the
network size and the length of its bud.  Such entries correspond to
edges not originally present in the stem-bud network whose addition
will have the greatest impact on performance.  We also establish that
edge-weight modifications in the stem do not affect the stability of
the resulting network.  In our third contribution, we consider
networks modified at a single edge with a given weight and provide
bounds on the value of the trace, trace inverse, and log-det of the
Gramian before and after modification. These bounds allow us to
estimate the global optima of these metrics under single-edge
modification, something we use in our numerical examples to verify the
efficacy of ECM in capturing the most relevant network edges.  We also
determine a sufficient condition on the amount of change in the edge
weight that ensure the network remains stable after modification. This
condition is valid for arbitrary stable networks and weights.
Finally, we illustrate our results in simulation on a family of 6-node
stem-bud networks and $1000$ random Erd\H{o}s-R\'{e}nyi networks.


\subsection*{Notation}
We let $\real$ and $\complex$ denote the set of real and complex
numbers, respectively. For $x\in \real$ (resp. $x\in \complex$),
$\vert x\vert$ denotes its absolute value (resp. magnitude). When
applied to a vector or matrix, the operation $| \cdot |$ is taken
elementwise.  For $j \in \until{n}$, $e_j \in \real^n$ is the $j^{th}$
canonical unit vector.  By $(\cdot)^{\top}$ we denote the transpose of
a vector or matrix, and by $\Vert \cdot \Vert$ its Frobenius norm.
Given a square matrix~$A$, we denote its trace, determinant, and
spectral radius by $\tr(A)$, $\det (A)$, and $\rho(A)$ resp., and its
$(i,j)^{th}$ element by $a_{ij}$.  We use $A\succeq(\succ)0$ to denote
that $A$ is a positive semi-definite (definite) matrix, and
$A_1\succeq A_2$ to denote that $A_1-A_2\succeq 0$.  For a symmetric
matrix~$A$, $\lambda_i(A)$ denotes its $i^{th}$ largest eigenvalue and
$\lambda_{\min}\left(A\right)$ its smallest one.  We use $I$ to denote
the identity matrix of appropriate dimensions.
%
%

\section{Problem Statement}\label{research-1}
Consider a network of $n$ nodes represented by the triplet
$\mathcal{G}_A = \left(\mathcal{V},\mathcal{E}_A,w_A\right)$, where
$\mathcal{V}=\lbrace 1,2,\ldots,n\rbrace$ is the node set,
$\mathcal{E}_A=\set{\left(i,j\right) \mid i\in \mathcal{V}, j\in
  \mathcal{V}}$ is the edge set, and $w_A:\mathcal{E}_A\mapsto \real$
is a weight function. The pair $(i,j)$ denotes an edge directed from
node $i$ to node $j$, i.e., $i\longrightarrow j$. The weighted
adjacency matrix $A=\left(a_{ji}\right)\in\real^{n\times n}$ is
defined by $a_{ji} = w_A (i,j)\neq0$ if $\left(i,j\right)\in
\mathcal{E}_A$, else $a_{ji}=0$. The network follows the discrete
linear time-invariant dynamics,
\begin{align}\label{peq-2.0.1}
  x\left(t+1\right)&=Ax\left(t\right)+Bu\left(t\right), \quad   t \in \lbrace0 ,\ldots,T-1\rbrace,
\end{align}
where $T>0$ is a finite time horizon, $x\in \real^n$ and $u\in
\real^m$ are the state and the input vectors respectively.  $B
=\begin{pmatrix} b_1 & b_2 & \cdots & b_i & \cdots & b_m\end{pmatrix}
\in \real^{n\times m}$ denotes the input matrix.
We assume $B$ is known and the pair $\left(A,B\right)$ is controllable
for $T=n$. The control input $u$ might correspond to a known input
specified by the designer, an unknown disturbance, or a malicious
input. Note that $A$ is stable if $\rho(A)<1$.

The controllability of \eqref{peq-2.0.1} refers to the ability to
steer the state from an initial condition $x\left(0\right)=x_0$ to any
arbitrary final condition $x\left(T\right)=x_T$ in $T$ steps by
appropriately selecting the control input sequence
$\{u\left(0\right),u\left(1\right),\ldots u\left(T-1\right)\}$.  The
controllability of~\eqref{peq-2.0.1} can be assessed in a number of
ways~\cite{CTC:98}. Here, we employ the controllability Gramian,
\begin{align}\label{peq-2.0.2}
  \Wc = \sum_{t=0}^{T-1}A^tBB^\top{A^t}^{\top}.
\end{align}
The system $\left(A,B\right)$ is controllable in $T$ steps if the Gramian 
$\Wc$ is positive definite.

Controllability is a qualitative property that, per se, does not
capture the input's energy effort required to actually steer the
system state. To address this point, one can employ controllability
metrics,
cf.~\cite{FP-SZ-FB:14,THS-FLC-JL:16,GY-GT-BB-JJS-YYL-ALB:15,SZ-FP:17},
based on the spectral properties of the Gramian such as $\tr(\Wc)$,
$-\tr(\Wc^{-1})$, $\det(\Wc)$, $\log\det(\Wc)$, and
$\lambda_{\min}(\Wc)$ to measure the system performance.  The
interpretation of all these metrics stems from the fact that the
minimum energy required to take a dynamical system from the $0-$state
to a desired final state $x_f$ in time $T$ is
$x_f^{\top}{\mathcal{W}}^{-1}x_f$.  For instance, $\tr(\Wcinv)$ has
the interpretation of the average control energy when $x_f$ is a
zero-mean, unit-variance random target state.  If $v_i$ is the
eigenvector corresponding to the $i^{th}$ eigenvalue $\lambda_i$ of
$\Wc$ then, $v_i^{\top} {\mathcal{W}}^{-1}v_i=\frac{1}{\lambda_i}$ is
the \textit{eigen-energy}, i.e., the minimum energy required to move
the system in the direction $v_i$.  For the smallest eigenvalue,
$\frac{1}{\lambda_{\min}}$ represents the energy required to steer the
system in the most difficult direction.
If one considers $y=x$ as the output, then $\tr(\Wc)$ is the square of
the $\mathcal{H}_2$-norm of the system and is related to the average
controllability in all directions of the state space.  This also
corresponds to the energy in the output response to an unit impulse
input and it is also the expected root mean square value of the output
response to a white noise excitation input~\cite{KZ-JD-KG:95}.
The metrics $\det(\Wc)$ and $\log\det(\Wc)$ are related to the volume
of the ellipsoid containing the set of states that can be reached by
one unit or less energy~\cite{THS-FLC-JL:16}.

Our goal is to study the effect of changes in the network structure on
its controllability properties while maintaining  the input
structure intact.  By changes in network structure, we mean modifications of
the weights of existing edges or addition of new edges of suitable
weight. This has applications in practical problems such as mitigating
the effect of malicious attacks at input nodes or network edges, or
suppressing output response at particular nodes caused by
malicious inputs.  Our analysis is motivated by two complementary
situations:
\begin{enumerate}
\item scenarios where we are interested in making the network more
  easily controllable with respect to the control input nodes,
\item scenarios where we seek to make a network more difficult to
  control with respect to malicious input nodes.
\end{enumerate}
It is possible that both the scenarios occur concurrently, where the
set of input nodes is a combination of known control input nodes and
malicious ones.  Mathematically, these scenarios can be formalized as
optimization problems where the objective function corresponds to one
of the Gramian-based performance metrics described above and the
decision variables correspond to network edge selection along with the
corresponding weight allocation, with suitable budget constraints.
Such optimization problems are nonlinear and non-convex, and hence
computationally challenging.  As the size of the network increases,
exhaustive search becomes impractical. Our aim here is to develop
formal tools to characterize the importance of each network edge and
its impact on the system performance metrics. Such tools could
potentially be paired with known optimization procedures (e.g., by
reducing the search space to the most significant edges) to address
the scenarios described above.

\section{Edge Centrality as a Measure of First-Order Impact on Network
  Performance}\label{math-1}
Here we study the effect of edge weight perturbation on the
performance metrics described in Section~\ref{research-1}. We start by
deriving an expression for the gradient of these metrics with respect
to edge weights. We then build on this result to introduce a novel
notion of edge centrality that aims to capture the importance of
individual edges in determining network behavior.


\subsection{Network First-Order Perturbation
  Analysis}\label{first-order}
To analyze the effect of perturbing edge weights on the performance
metrics, we start by studying the first-order variation, i.e., the
gradient of a generalized version of the Gramian matrix with respect
to the elements of the adjacency matrix.

\begin{theorem}\longthmtitle{Gradient of scalar function of
    generalized Gramian with respect to edge weights}\label{thm1}
  Consider $A\in \real^{n\times n}$, a symmetric matrix $P\in
  \real^{n\times n}$ and $H = \begin{pmatrix} h_1 \ldots h_k \ldots
    h_m
  \end{pmatrix} \in \real^{n\times m}$. For $\phi\left(A\right)=
  \sum_{t=0}^{T-1}A^tHH^{\top}{A^t}^{\top}$ then
  \begin{align}\label{peq-thm1-1}
    \frac{\partial }{\partial
      a_{ji}}\tr\left(P\phi\left(A\right)\right) &=
    2\sum_{k=1}^{m}\sum_{t=1}^{T-1}
    \tr \left(\overline{C}_k^{\left(t\right)}
      \overline{O}_k^{\left(t\right)}e_ie_j^{\top}\right)
    \notag
    \\
    &= 2\sum_{t=1}^{T-1}
    \tr \left(\overline{C}_H^{\left(t\right)}
      \overline{O}_H^{\left(t\right)}e_ie_j^{\top}\right),
  \end{align}
  where
  \begin{subequations}\label{peq-thm1-2}
    \begin{align}
       \overline{C}_k^{\left(t\right)} &= \begin{pmatrix}
        {A^{t-1}}^{\top}PA^th_k & \cdots & A^{\top}PA^th_k & PA^th_k
      \end{pmatrix},
      \\
      \overline{O}_k^{\left(t\right)}&=\begin{pmatrix}
        h_k & Ah_k & \cdots & A^{t-2}h_k & A^{t-1}h_k
      \end{pmatrix}^{\top},\\
       \overline{C}_H^{\left(t\right)} &= \begin{pmatrix}
        {A^{t-1}}^{\top}PA^tH & \cdots & A^{\top}PA^tH & PA^tH
      \end{pmatrix},
      \\
      \overline{O}_H^{\left(t\right)}&=\begin{pmatrix} H & AH & \cdots
        & A^{t-2}H & A^{t-1}H
      \end{pmatrix}^{\top}.
  \end{align}
\end{subequations}
\end{theorem}
\begin{proof}
  The derivative of $\phi$ with respect to an edge weight can be
  expressed as 
  \begin{align*}
    \frac{\partial \phi}{\partial a_{ji}} &=
    \sum_{t=1}^{T-1}\left[C_j^{\left(t\right)}O_i^{\left(t\right)}HH^{\top}{A^t}^{\top}
      +
      A^tHH^{\top}\left(C_j^{\left(t\right)}O_i^{\left(t\right)}\right)^{\top}\right],
  \end{align*}
  where we have used \cite[Theorem 3.1]{PVC-EN-JC:21-tnse} and the notation
  \begin{align}
  	\label{peq-thm1-3a}
    C_j^{\left(t\right)} & = \begin{pmatrix} A^{t-1}e_j & A^{t-2}e_j &
      \cdots & Ae_j & e_j
    \end{pmatrix},
    \\
    O_i^{\left(t\right)}&=\begin{pmatrix} e_i & A^{\top}e_i & \cdots &
      {A^{t-2}}^{\top}e_i & {A^{t-1}}^{\top}e_i
    \end{pmatrix}^{\top}.
  \end{align}
  Now using $HH^{\top}=\sum_{k=1}^{m}h_kh_k^{\top}$ and
  $\frac{\partial}{\partial a_{ji}}\tr\left(P\phi\left(A\right)\right)
  = \tr\left(P\frac{\partial\phi\left(A\right)}{\partial
      a_{ji}}\right)$ gives,
  \begin{align*}
    \notag\frac{\partial }{\partial
      a_{ji}}\tr\left(P\phi\left(A\right)\right) &=
    \tr\Bigg\{\sum_{k=1}^{m}\sum_{t=1}^{T-1}
    \left[PC_j^{\left(t\right)}
      O_i^{\left(t\right)}h_k{h_k}^{\top}{A^t}^{\top}\right.
    \\
    &\hspace*{0.5cm}\left.+PA^th_k{h_k}^{\top}
      \left(C_j^{\left(t\right)}O_i^{\left(t\right)}\right)^{\top}\right]\Bigg\}.
  \end{align*}
  Using cyclic permutations of matrices and $\tr
  \left(XY^{\top}\right) = \tr\left(X^{\top}Y\right)$,
  \begin{align}\label{peq-thm1-6}
    \frac{\partial}{\partial
      a_{ji}}\tr\left(P\phi\left(A\right)\right) &=
    2\sum_{k=1}^{m}\sum_{t=1}^{T-1}
    \tr\left(O_i^{\left(t\right)}
      h_kh_k^{\top}{A^t}^{\top}PC_j^{\left(t\right)}\right).
  \end{align}
  Consider
  \begin{align}\label{peq-thm1-7}
 O_i^{\left(t\right)}h_k&=\begin{pmatrix} h_k^{\top}e_i &
      h_k^{\top}A^{\top}e_i & \cdots & h_k^{\top} {A^{t-1}}^{\top}e_i
    \end{pmatrix}^{\top}
    \\
    \notag &=\begin{pmatrix} e_i^{\top}h_k & e_i^{\top}Ah_k & \cdots &
      e_i^{\top} Ah_k
    \end{pmatrix}^{\top}
    =\overline{O}_k^{\left(t\right)}e_i.
  \end{align}
  Similarly using $P=P^{\top}$,
  \begin{align}\label{peq-thm1-8}
    h_k^{\top}{A^t}^{\top}PC_j^{\left(t\right)} &= \begin{pmatrix}
      h_k^{\top}{A^t}^{\top}PA^{t-1}e_j & \cdots &
      h_k^{\top}{A^t}^{\top}Pe_j
    \end{pmatrix}
    \\
    &= \begin{pmatrix} e_j^{\top}{A^{t-1}}^{\top}PA^th_k & \cdots &
      e_j^{\top}PA^th_k
    \end{pmatrix}
    \notag =e_j^{\top}\overline{C}_k^{\left(t\right)}.
  \end{align}
  Using \eqref{peq-thm1-7} and \eqref{peq-thm1-8} in
  \eqref{peq-thm1-6} yields
  \begin{align*}
    \frac{\partial}{\partial
      a_{ji}}\tr\left(P\phi\left(A\right)\right) &=
    2\sum_{k=1}^{m}\sum_{t=1}^{T-1}
    \tr\left(\overline{O}_k^{\left(t\right)}e_ie_j^{\top}
      \overline{C}_k^{\left(t\right)}\right),
    \\
    &=2\sum_{k=1}^{m}\sum_{t=1}^{T-1} \tr
    \left(\overline{C}_k^{\left(t\right)}\overline{O}_k^{\left(t\right)}e_ie_j^{\top}\right).
  \end{align*}
  The result now follows from the expressions in~\eqref{peq-thm1-2}.
\end{proof}

One can interpret Theorem \ref{thm1} as an extension of the result
in~\cite[Theorem 3.1]{PVC-EN-JC:21-tnse} for the gradient of the
Gramian with respect to an edge weight.  The presence of the arbitrary
symmetric matrix $P$ in~\eqref{peq-thm1-1} provides greater
versatility, and in fact leads to a unified way of computing the
gradients of the various Gramian-based performance metrics, as we
show next.



\begin{corollary}\longthmtitle{Gradient of performance metrics}\label{thm2}
  Consider the network dynamics~\eqref{peq-2.0.1}.
  For $P\in \real^{n\times n}$ symmetric, let
  \begin{align}
    \label{peq-thm2-1}
    \Theta_P=\sum_{k=1}^{m}\sum_{t=1}^{T-1}\overline{C}_k^{\left(t\right)}
    \overline{O}_k^{\left(t\right)}
  \end{align} 
  with $\overline{C}_k^{\left(t\right)}$,
  $\overline{O}_k^{\left(t\right)}$ defined in Theorem~\ref{thm1} with
  $H=B$.  Then,
  \begin{enumerate}
  \item $\frac{\partial }{\partial
      a_{ji}}\tr\left(\mathcal{W}\right)=2\Theta_I\left(j,i\right)$;
  \item $\frac{\partial }{\partial a_{ji}}\log\det\left(
      \mathcal{W}\right)=2\Theta_{\Wc^{-1}}\left(j,i\right)$;
  \item $\frac{\partial }{\partial a_{ji}}\{-\tr\left(
      {\mathcal{W}}^{-1}\right)\}=2\Theta_{\Wc^{-2}}\left(j,i\right)$;
  \item $\frac{\partial }{\partial a_{ji}}\lambda_i\left(
      \mathcal{W}\right)=2\Theta_{V_i}\left(j,i\right)$, where $V_i
    =v_iv_i^{\top}$ and $v_i$ is the eigenvector of $\mathcal{W}$
    corresponding to the $i^{th}$ largest eigenvalue
    $\lambda_i\left(\mathcal{W}\right)$.
  \end{enumerate}
\end{corollary}
\begin{proof}
  We prove the result by making repeated use of Theorem~\ref{thm1}
  with $\phi$ defined by the choice $H=B$ and considering different
  $P$ as necessary.  We make use of the following matrix
  properties. For $X$, $Y$, $Z$,
  \begin{align*} 
    \tr (Xe_ie_j^{\top}) &= X\left(j,i\right)
    \\
    \tr\left(XYZ\right) & = \tr\left(ZXY\right) =\tr\left(YZX\right) ,
    \\
    \frac{\partial}{\partial y_{ji}}\tr\left(XY\right) &=
    \tr\left(X\frac{\partial Y}{\partial y_{ji}}\right).
  \end{align*}
  Now, the proof of each item follows by combining these facts with
  the following choices. Case \emph{(i)} follows readily by
  considering $P=I$.  For case \emph{(ii)}, from \cite[Appendix
  A]{SB-LV:09},
  \begin{align*}
    \frac{\partial }{\partial a_{ji}} \log\det\left(\mathcal{W}\right)
    = \tr\left({\mathcal{W}}^{-1}\frac{\partial\mathcal{W}}{\partial
        a_{ji}}\right) .
  \end{align*}
  This case then follows by looking at $P={\mathcal{W}}^{-1}$ as a
  constant matrix that does not change with~$A$.  Similarly, for case
  \emph{(iii)}, from \cite[Appendix A]{SB-LV:09},
  \begin{align*}
    \frac{\partial }{\partial
      a_{ji}}\{-\tr\left({\mathcal{W}}^{-1}\right)\} =
    \tr\left({\mathcal{W}}^{-1}\frac{\partial\mathcal{W}}{\partial
        a_{ji}}{\mathcal{W}}^{-1}\right)=\tr\left({\mathcal{W}}^{-2}\frac{\partial
        \mathcal{W}}{\partial a_{ji}}\right).
  \end{align*}
  This case then follows by looking at $P={\mathcal{W}}^{-2}$ as a
  constant matrix that does not change with~$A$.  Finally, for case
  \emph{(iv)}, the eigenvalue equation, $\mathcal{W}v_i =
  \lambda_i\left(\mathcal{W}\right) v_i$. As $\mathcal{W}$ is real and
  symmetric, we use the fact that $v_i^{\top}v_i=1$ repeatedly to
  first express $\lambda_i\left(\mathcal{W}\right) =
  v_i^{\top}\mathcal{W}v_i$ and then
  \begin{align*}
    \frac{\partial }{\partial a_{ji}}\lambda_i\left(
      \mathcal{W}\right) = v_i^{\top}\frac{\partial
      \mathcal{W}}{\partial a_{ji}}v_i =
    \tr\left(v_iv_i^{\top}\frac{\partial \mathcal{W}}{\partial
        a_{ji}}\right).
  \end{align*}
  This case then follows by looking at $P=v_iv_i^{\top}=V_i$ as a
  constant matrix that does not change with~$A$.
\end{proof}

Note that Corollary~\ref{thm2}(iv) with $i=n$ corresponds to the
smallest eigenvalue of the Gramian~$\Wc$.

\begin{remark}\longthmtitle{Computational effort in gradient
    computation}
  {\rm Corollary~\ref{thm2} (i)-(iii) extend to graphs with arbitrary
    edge weights the result in our previous work~\cite[Corollary
    3.2]{PVC-EN-JC:21-tnse}, which provides different, equivalent
    expressions for digraphs with positive edge weights. There is an
    additional key difference in the computation requirements of each
    expression, which are significantly lighter here.  To compute the
    gradient with respect to an edge weight, we need to calculate for
    $t \in \until{T-1}$ the matrices $\overline{C}_k^{(t)}$,
    $\overline{O}_k^{(t)}$ given in \eqref{peq-thm1-2} according to
    Corollary~\ref{thm2}, and the matrices $C_j^{(t)}$, $O_i^{(t)}$ in
    \eqref{peq-thm1-3a} according to~\cite[Corollary
    3.2]{PVC-EN-JC:21-tnse}. Both sets of computations require
    approximately the same effort.  However, with the approach here,
    the summation over $k$ from $1$ to $m$ (which is independent of
    $n$) yields the gradient of the desired performance matrix with
    respect to all $n^2$ edge weights.  Instead, to obtain the latter
    with the approach in~\cite{PVC-EN-JC:21-tnse}, one has to perform
    a computation like this $n^2$ times (by considering the
    combinations across $i \in \until{n}$ and $j \in \until{n}$),
    resulting in a significantly higher computational effort as $n$
    increases.  
  } \oprocend
\end{remark}

Finally, we note here that the expression~\eqref{peq-2.0.2} of the
Gramian, along with Theorem~\ref{thm1} and Corollary~\ref{thm2}, are
valid regardless of the stability of~$A$. When considering stable
networks, in the limit $T \rightarrow \infty$, the controllability
Gramian $\Wci$ becomes the solution of the discrete-time Lyapunov
equation,
\begin{align}\label{peq-2.0.2a}
  A\Wci A^{\top}-\Wci+BB^{\top} = 0.
\end{align}
In fact, this equation has a valid solution only if $A$ is stable,
cf.~\cite{CTC:98}. As the time horizon $T$ grows to infinity, one has
$\Vert \Wci-\Wc\Vert\longrightarrow 0$.  For such cases, one could
compute the gradient of the Gramian with respect to edge weights by
differentiating~\eqref{peq-2.0.2a} to obtain
\begin{align}\label{peq-2.0.2b}
  A\frac{\partial \Wci}{\partial a_{ji}}A^{\top}-\frac{\partial
    \Wci}{\partial a_{ji}}+\frac{\partial A}{\partial a_{ji}}\Wci
  A^{\top}+A\Wci\left(\frac{\partial A}{\partial
      a_{ji}}\right)^{\top}=0 .
\end{align}
To compute the gradient of any performance metric, we need to
solve~\eqref{peq-2.0.2a} and~\eqref{peq-2.0.2b} for $n^2$ number of
edges.
In contrast, Corollary \ref{thm2} offers a computationally efficient
way to compute the gradient of various performance metrics and given
input set using simple matrix multiplications irrespective of the
stability of~$A$.

\subsection{The Edge Centrality Matrix}\label{ECM1}

Here we introduce the edge centrality matrix as a way of capturing the
importance of the network connections in driving its behavior and
performance.  Corollary \ref{thm2} shows that, with an appropriate
choice of the symmetric matrix $P$, the gradient of various
performance metrics can be expressed by means of $\Theta_P$
in~\eqref{peq-thm2-1}.  Given a performance metric, we refer to
\begin{align*}
  \Theta_P^k = \sum_{t=1}^{T-1}\overline{C}_k^{\left(t\right)}
  \overline{O}_k^{\left(t\right)} ,
\end{align*}
as the \emph{edge centrality matrix associated to the $k^{th}$-input}
and~to
\begin{align}\label{eq:additivity}
  \Theta_P = \sum_{k=1}^{m}\Theta_P^k ,
\end{align}
as the \emph{edge centrality matrix} (ECM) for the network
dynamics~\eqref{peq-2.0.1}.  Given the energy interpretations
associated to these notions, the ECMs exactly encode the first-order
changes in physically realizable quantities whenever a network
structure is perturbed.  The additivity property reflected
in~\eqref{eq:additivity} is particularly noteworthy, because it
captures the specific contribution to each edge's centrality of the
presence of actuator~$k \in \until{m}$ in the network. This offers the
system designer flexibility to examine the effects of edge
perturbation due to individual inputs, a subset of inputs, or the
complete set of inputs.  Both properties, the one-to-one
correspondence with first-order changes in the performance metric and
the ability to pinpoint the impact of each actuator, are
significant advantages of the ECM concept over the edge centrality
metric proposed in our previous work~\cite{PVC-EN-JC:21-tnse}, whose
construction is also based on the controllability Gramian.

\section{Edge Centrality Matrix of Stem-Bud
  Networks}\label{stem-bud-1}

We are interested in characterizing the structure of the ECM. The
complexity of this goal is daunting for general networks, so here we
focus our attention on the particular class of directed stem-bud
networks, cf.~\cite{SZ-FP:17}. These networks possess a diagonal
controllability Gramian, which significantly facilitates the study of
their ECM.

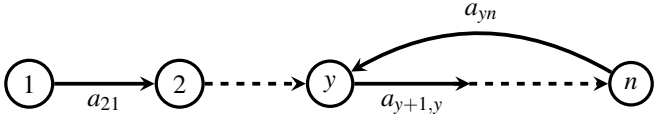
\begin{figure}[htb]
  \centering
  \begin{tikzpicture}[>=stealth,line width=0.43mm]
    \node[shape=circle,draw=black,line width=0.43mm] (no1) at (0,0){$1$};
    \node[shape=circle,draw=black,line width=0.43mm] (no2) at (2,0){$2$};
    \node[shape=circle,draw=black,line width=0.43mm] (no3) at (4,0){$y$};
 \node[] (no4) at (6,0){};
 \node[] (no4a) at (5.5,0){};
    \node[shape=circle,draw=black,line width=0.43mm] (no5) at (8,0){$n$};
    \draw[->,line width = 0.5mm](no1)--node[below]{$a_{21}$}(no2);
    \draw[dashed,->,line width = 0.5mm](no2)--(no3);
    \draw[->,line width = 0.5mm](no3)--node[below]{$a_{y+1,y}$}(no4);
    \draw[dashed,->,line width = 0.5mm](no4a)--(no5);
    \draw[->,line width=0.5mm](no5) edge[bend right] node [above] {$a_{yn}$}(no3);
  \end{tikzpicture}
  \caption{A directed stem-bud network with $n$ nodes and $y$ as the
    junction.}\label{5-node-2}
\end{figure}

A directed stem-bud network is a combination of a directed line
network, the \emph{stem} from node $1$ to node $y$, with a directed
ring network, the \emph{bud} starting and ending at node $y$, see
Fig. \ref{5-node-2}.  Hence, the stem contains the node sequence
$1\rightarrow 2 \rightarrow \ldots \rightarrow y$ and the bud contains
the node sequence $y\rightarrow y+1\rightarrow \ldots \rightarrow n
\rightarrow y$. The node $y$ is called the \emph{junction}.
Consequently, the possible non-zero entries in the weighted adjacency
matrix of a stem-bud network are $a_{i,i-1}$, with $i=2,3,\ldots,n$,
and $a_{yn}$.  Note that a stem-bud network is simply a directed line
network if $a_{yn}=0$
%
and a directed ring network if $y=1$.  The network is controllable
with only one actuator if it is placed at node~$1$. When multiple
actuators are used, node $1$ should be actuated to achieve system
controllability.  For $1\leq y \leq n-1$, we let $L_b=n-y+1$ denote
the length of the bud and
$\Lambda_{b}=a_{yn}\prod_{j=y+1}^{n}a_{j,j-1}$ its contribution.  For
convenience, in the directed line network case, we make the convention
that $y=0$ and $L_b=\infty$.

The next result establishes that directed stem-bud networks have
diagonal controllability Gramians and is a generalization to arbitrary
time horizons~$T$ of the result in~\cite[Section 3]{SZ-FP:17} for the
case with $T=n$.
 
\begin{proposition}\longthmtitle{Controllability Gramian of directed
    stem-bud networks}\label{prop1}
  Consider a directed stem-bud network without self-loops and input at
  node $i_b \in \until{n}$. Let $T >0$ be the time horizon. Then, the
  controllability Gramian $\mathcal{W}_{i_b}$ is diagonal.
\end{proposition}
\begin{proof}
  Let $b=e_{i_b}$ and $v_t=A^{t-1}b\in\real^n$ for $1\leq t \leq
  T$. Note that, at any given time $t$, the input reaches only one
  node, say $p$. Consequently $v_t$ is a vector with only one non-zero
  element, denoted $v_t\left(p\right)$, at the $p^{th}$ component. We
  next provide expressions for it, distinguishing between two cases,
  $1\leq T\leq n-i_b+1$ and $T> n-i_b+1$,
  \begin{enumerate}
  \item For $1\leq T\leq n-i_b+1$, one has $p=i_b+t-1$ and
    \begin{align}
      \label{peq-prop1-1}
      &v_t(p) = \prod_{j=k+1}^{p}a_{j,j-1}\;\text{for}\; t \in \{1,2,\ldots,T\}.
    \end{align}
  \item For $T>n-i_b+1$, if $1\leq t \leq n-i_b+1$, then $p$ and
    $v_t\left(p\right)$ are as in~(i). If $t> n-i_b+1$, let
    $t-\left(n-i_b+2\right) = \xi L_b+\zeta$, i.e., $\xi$ and $\zeta$
    are the quotient and remainder, respectively, of dividing
    $t-\left(n-i_b+2\right)$ by $L_b$. To determine the value of $p$,
    note that it takes $n-i_b$ hops to reach node $n$ from node $i_b$,
    and then one more hop to reach node $y$.  Of the remaining
    $t-\left(n-i_b+2\right)$ hops left, after traversing the loop
    $\xi$ times, we reach the $\zeta$th node on the loop, and
    consequently $p=y+\zeta$. The expression for $v_t\left(p\right)$
    is then
    \begin{align}
      \label{peq-prop1-2}
      v_t(p) =
      \begin{cases}
        a_{yn}\Lambda_{b}^{\xi}\prod_{j=i_b+1}^{n}a_{j,j-1}\prod_{i=y+1}^{p}a_{i,i-1}
        & \text{for} \; \zeta>0,
        \\
        a_{yn}\Lambda_{b}^{\xi}\prod_{j=i_b+1}^{n}a_{j,j-1}&
        \text{for} \; \zeta=0.
      \end{cases}
    \end{align}  
  \end{enumerate}
  Consequently,
  $\mathcal{W}_{i_b}=\sum_{t=0}^{T-1}A^tbb^{\top}{A^t}^{\top} =
  \sum_{t=1}^{T}v_tv_t^{\top}$ is a diagonal matrix.
\end{proof}

For undirected stem-bud networks, the controllability Gramian might in
general not be diagonal.  As the Gramian is an additive function in
the input location space, cf. \cite{THS-FLC-JL:16}, in case of
multiple inputs, one readily deduces from Proposition~\ref{prop1} that
the Gramian is also diagonal. Building on this result, we next
characterize the structure of the ECM of a directed stem-bud
network. In our next result, an element $q_{ji}$ of matrix $Q\in
\real^{n\times n}$ belongs to the $(i-j)$th super-diagonal if $j<i$,
to the $(j-i)$th sub-diagonal if $j>i$, and to the main diagonal
if~$j=i$.

\begin{theorem}\longthmtitle{Structure of ECM of stem-bud
    networks}\label{thm31}
  Consider a controllable $n-$node directed stem-bud network with
  dynamics \eqref{peq-2.0.1} and $m$ inputs.  For the choices of
  matrix $P \in \real^{n \times n}$ specified in Corollary~\ref{thm2},
  the ECM $\Theta_P$ may have non-zero elements in
  \begin{enumerate}
  \item the set of sub-diagonals $N_{\sub}=\{1 + i\, L_b: i \in
    \ontil{k_{\sub}}\}$, 
    with $k_{\sub}=\big\lfloor \frac{n-2}{L_b}\big\rfloor$;
  \item the set of super-diagonals $N_{\sup}=\{i \, L_b-1 : i \in
    \until{k_{\sup}} \}$,
    with $k_{\sup}=\big\lfloor \frac{n}{L_b}\big\rfloor$.
  \end{enumerate}
\end{theorem}
\begin{proof}
  Note that since $\Theta_P$ is additive in the input space,
  $\Theta_P=\sum_{i_b=1}^{m}\Theta_P^{i_b}$,
  cf.~\eqref{eq:additivity}, it is enough to reason for the
  single-input case.  Consider then the case of an input located at
  node $i_b$, i.e., with $b_{i_b}=e_{i_b}$.  Note that $A^te_{i_b}$ is
  a column vector with only one non-zero value at the $t+{i_b}$
  coordinate.  From~\eqref{peq-thm2-1}, $\Theta_P^{i_b}$ consists of
  additive terms of
  $\overline{C}_{i_b}^{\left(t\right)}\overline{O}_{i_b}^{\left(t\right)}$
  running in $t\in \{1,\ldots,T-1\}$.  We consider three cases as
  follows.
	
  {\textit{Case 1}} ($t < y-{i_b}$):  Relying on~\eqref{peq-thm1-2},
  consider one general term in the product
  $\overline{C}_{i_b}^{\left(t\right)}\overline{O}_{i_b}^{\left(t\right)}$
  as ${A^s}^{\top}PA^te_{i_b}{e_{i_b}}^{\top}{A^{t-1-s}}^{\top}$,
  where $s$ is an integer and $0\leq s \leq t-1$.  In this case, the
  only non-zero entry of $A^te_{i_b}$ corresponds to a node $t+{i_b}$
  at the stem since $t < y-{i_b}$.  From Proposition~\ref{prop1}, the
  Gramian $\mathcal{W}$ is diagonal, and hence the matrix $P$ in
  Corollary \ref{thm2} is diagonal, and may be singular or
  non-singular.  If $PA^te_{i_b}$ is non-zero, then it retains the
  same structure as $A^te_{i_b}$. In that case, the quantity
  ${A^s}^{\top}PA^te_{i_b}$ is a column vector with only one non-zero
  element at the $t+{i_b}-s$ coordinate. Now,
  ${e_{i_b}}^{\top}{A^{t-1-s}}^{\top}$ results in a row vector with
  only non-zero element at the $t+{i_b}-s-1$ coordinate. So,
  ${A^s}^{\top}PA^te_{i_b}{e_{i_b}}^{\top}{A^{t-1-s}}^{\top}$ is a
  matrix with only one non-zero element at the
  $\left(t+{i_b}-s,t+{i_b}-s-1\right)$ position, which is part of the
  first sub-diagonal.

  {\textit{Case 2}} ($y-{i_b}\leq t \leq n-{i_b}$): In this case,
  ${A^s}^{\top}PA^te_{i_b}$ results in the input reaching nodes
  $t+{i_b}-s$, $t+{i_b}-s+\zeta L_b$, where $\zeta \in
  \{1,2,\ldots,\zeta_{max}\}$ on the stem and node
  $\left(n-s+t+{i_b}-y+\zeta_{max}L_b+1\right)$ on bud (here, 
  $\zeta_{max}=\big\lfloor \frac{s-t-{i_b}+y}{L_b}\big\rfloor$). Now,
  ${e_{i_b}}^{\top}{A^{t-1-s}}^{\top}$ results in a row vector with
  only non-zero coordinate at $t+{i_b}-s-1$ place. So,
  ${A^s}^{\top}PA^te_{i_b}{e_{i_b}}^{\top}{A^{t-1-s}}^{\top}$ is a
  matrix with non-zero elements at positions,
  \begin{enumerate}[label=(\roman*)]
  \item $\left(t+{i_b}-s,t+{i_b}-s-1\right)$ i.e., first sub-diagonal.
  \item $\left(t+{i_b}-s+\zeta L_b,t+{i_b}-s-1\right)$ i.e.,
    $\left(\zeta L_b+1\right)$ sub-diagonal for $\zeta\in
    \{1,2,\ldots,\zeta_{max}\}$.
  \item $\left(n-s+t+{i_b}-y+\zeta_{max}L_b+1,t+{i_b}-s-1\right)$
    i.e.,
    $n-y+\zeta_{max}L_b+2=\left(\left(\zeta_{max}+1\right)L_b+1\right)$
    sub-diagonal.
  \end{enumerate} 
  
  {\textit{Case 3}} ($t > n-{i_b}$): Proceeding as in Case 2,
  $A^te_{i_b}$ will make the input reach node $y+r_1$ where $r_1 =
  t-y+{i_b}-\delta_1L_b$ and $\delta_1= \big\lfloor
  \frac{t-y+{i_b}}{L_b}\big\rfloor$. ${A^s}^{\top}PA^te_{i_b}$ will
  result in nodes $y-s+r_1+\delta_2L_b$ where $\delta_2 =
  \{0,1,2\ldots,\delta_2^{max}\}$, $\delta_2^{max}=\big\lceil
  \frac{1+s-y-r}{L_b}\big\rceil$ on the stem. On the bud, the node is
  $n-r_2$ with $r_2 = s-r_1-1-\delta_3L_b$ and $\delta_3=\big\lfloor
  \frac{s-r-1}{L_b}\big\rfloor$. ${e_{i_b}}^{\top}{A^{t-1-s}}^{\top}$
  is the node $y+r_3$ where $r_3=t-1-s-y+{i_b}-\delta_4L_b$ and
  $\delta_4=\big\lfloor \frac{t-1-s-y+{i_b}}{L_b}\big\rfloor$. So,
  ${A^s}^{\top}PA^te_{i_b}{e_{i_b}}^{\top}{A^{t-1-s}}^{\top}$ is a
  matrix with non-zero elements at positions,
  \begin{enumerate}[label=(\roman*)]
  \item
    $\left(t+{i_b}-s-\delta_1L_b+\delta_2L_b,t-1+{i_b}-s-\delta_4L_b\right)$
    i.e., $\left(\left(\delta_2+\delta_4-\delta_1\right)L_b+1\right)$
    sub-diagonal or
    $\left(\left(\delta_1-\delta_2-\delta_4\right)L_b-1\right)$
    super-diagonal.
  \item
    $\left(n+t+{i_b}-y-s-\delta_1L_b+1+\delta_3L_b,t-1+{i_b}-s-\delta_4L_b\right)$
    i.e.,
    $\left(\left(\delta_3+\delta_4+1-\delta_1\right)L_b+1\right)$
    sub-diagonal or
    $\left(\left(\delta_1-\delta_3-1-\delta_4\right)L_b-1\right)$
    super-diagonal.
  \end{enumerate} 	
  In all the three cases, we find the non-zero elements belong to the
  sub-diagonals $1,L_b+1,2L_b+1,\ldots,k L_b+1$, with
  $k L_b+1 \leq n-1$, or the super-diagonals
  $L_b-1,2L_b-1,\ldots, k L_b-1$, with $k L_b-1 \leq n-1$, as stated.
\end{proof}

Theorem~\ref{thm31} states that for stem-bud networks, the structure
of ECM depends only on the size $n$ of the network and the length
$L_b$ of its bud.  In general, the performance of the network may be
improved by suitably modifying its edges and their corresponding
weights. However, determining such modifications is a computationally
intensive process. The edge centrality matrix offers a criterium to
guide the selection of edges to modify by identifying those with
higher first-order impact on network performance.  All edges
corresponding to sub-diagonal elements (except those in the first
sub-diagonal) and all super-diagonal elements (except the element at
$(y,n)$) of ECM do not exist in the stem-bud network and are therefore
new.  Edges corresponding to an element in a sub-diagonal run along
the stem.  Edges corresponding an element in a super-diagonal form a
new cycle in the network.

When a stem-bud network is either a directed line network ($y=0$)
or a directed ring network ($y=1$), the next result shows that the
structure of their respective ECMs is the same as that of the weighted
adjacency matrix~$A$.

\begin{corollary}\longthmtitle{Structure of ECM of directed
    line and ring networks}\label{corr1}
  For directed line and ring networks, the edge centrality matrix
  $\Theta_P$ has the same structure as the weighted adjacency
  matrix~$A$.
\end{corollary}
\begin{proof}
  A directed line network corresponds to a stem-bud network with
  $L_b=\infty$.  From Theorem~\ref{thm31}, $k_{\sup} = 0$, meaning
  that $\Theta_P$ has non-zero elements in the first sub-diagonal
  only. This structure is the same as that of~$A$. A directed ring
  network corresponds to a stem-bud network with $L_b=n$, which yields
  $k_{\sub}=0$ and $k_{\sup}=1$. Thus $\Theta_P$ has non-zero elements
  in the first sub-diagonal and the $n-1$ super-diagonal only, i.e.,
  $\Theta_P$ has the same structure as that of~$A$.
\end{proof}

\section{Network Performance and Stability Bounds}\label{section4}
Here, we examine two complementary aspects regarding network
performance and stability motivated by our observations in
Section~\ref{stem-bud-1}.  On  one hand, we seek to bound the
impact on network performance that edge modification might have and to
understand  to what extent ECM is a good indicator of it. Given the
challenges in addressing this question, here we focus on quantifying
the impact over an infinite time horizon caused by the modification of
a single edge. On the other hand, to ensure that edge modifications do
not result in network instability, we characterize bounds on the weight
changes for each edge such that the network retains its stability
properties.

\subsection{Bounding the Change in Network Performance of Single-Edge
  Modification}\label{global-1}
Our interest here lies in quantifying the effect of edge perturbations in
stable network systems.
We consider the trace, log-det and, trace inverse of the Gramian as
performance metrics and provide bounds for each of them for the
original (unmodified) network as well as for the network modified at
one edge. Consistent with the interpretation of these metrics
regarding network controllability, we derive upper bounds for the
trace and log-det and a lower bound for the trace inverse. The
following notation is useful in our forthcoming discussion. For $w \in
\real$,  let
\begin{align}\label{defs-1}
  \Xc & = (I-\vert A\vert)^{-1} , \quad \alpha_{ij} =\frac{\vert
    w\vert}{1-\vert w\vert e_i\T\Xc e_j}, \quad \alpha =
  \underset{i\neq j}{\max}\; \alpha_{ij}, \notag
  \\
  \beta & =\max\Big\{\underset{i\neq j}{\max}\; 2e_i\T \Xc e_j,
  \underset{i\neq j}{\max}\; e_i\T \Xc e_j+\underset{j}{\max}\;
  \Vert\Xc e_j\Vert\Big\},
  \\
  \notag \gamma &=\underset{k}{\max}\;e_k\T\Xc\T\Xc e_k, \quad
  \bar{\gamma}=\underset{k}{\max}\;e_k\T\Xc\BB \Xc\T e_k.
\end{align}
We are ready to state the first result of this section.

\begin{theorem}\longthmtitle{Upper bound on trace of
    Gramian}\label{thm4}
  Given a network with adjacency matrix $A$ and $\rho(\vert
  A\vert)<1$, consider the modified network resulting from adding the
  weight $w \in \real$ to an edge $i\longrightarrow j$, $i\neq j$,
  such that $\rho(\vert A\vert +\vert w\vert e_je_i\T)<1$. Then,
  \begin{subequations}\label{peq-thm4-00}
    \begin{align}
      \label{peq-thm4-00a} \tr(\Wc_A) & \leq\tr(\Wci_A)  \leq
      \tr(\Hc_{\Xc}),
      \\
      \label{peq-thm4-00b} \tr(\Wc_{A+w e_je_i\T})
      & \leq
      (1+\alpha\beta)\tr(\Hc_{\Xc})+\alpha^2\gamma\bar{\gamma},
    \end{align}
  \end{subequations}
  where $\Hc_{\Xc}=\mathcal{X}\vert B\vert \vert
  B\vert^{\top}\mathcal{X}^{\top}$.
\end{theorem}
\begin{proof}
  To prove \eqref{peq-thm4-00a}, note that since $
  \mathcal{W}_A^{\infty}=
  \sum_{t=0}^{\infty}A^tBB^{\top}{A^t}^{\top}$, one has
  $\Wci_A\succeq\Wc_A$. As $\vert A\vert \geq A$ and $\vert B\vert
  \geq B$,
  \begin{align}
    \label{peq-lem3-0}
   \tr\Big(\sum_{t=0}^{\infty}\vert A
          \vert^t\vert B\vert \vert B\vert^{\top}{\vert A
            \vert^t}^{\top}\Big)\geq
    \tr(\mathcal{W}_A^{\infty})\geq \tr(\Wc_A).
  \end{align}	
  As $\vert A \vert \geq 0$ and $\rho(\vert A \vert)<1$, from
  \cite[Lemma 2.3.3]{GHG-CFVL:13}, $\mathcal{X}$ is non-singular and
  \begin{align}
    \label{peq-lem3-1}
    \mathcal{X} = \sum_{t=0}^{\infty}\vert A \vert^t.
  \end{align}
  Now consider,
  \begin{align}
    \label{peq-lem3-2}
    \mathcal{X}\vert B\vert \vert B\vert^{\top}\mathcal{X}^{\top} =
    \sum_{t=0}^{\infty}\vert A \vert^t\vert B\vert \vert
    B\vert^{\top}{\vert A \vert^t}^{\top}+\Psi(\vert A \vert,\vert
    B\vert),
  \end{align}
  where $\Psi(\vert A\vert,\vert B\vert)$ is a function of cross-terms
  involving $\vert A \vert$ and $\vert B \vert$. Taking trace
  throughout in \eqref{peq-lem3-2}, we obtain $\tr(\mathcal{X}\vert
  B\vert \vert B\vert^{\top}\mathcal{X}^{\top}) \geq \tr\Big(
  \sum_{t=0}^{\infty}\vert A \vert^t\vert B\vert \vert
  B\vert^{\top}{\vert A \vert^t}^{\top}\Big)$, which combined with
  \eqref{peq-lem3-0} gives us \eqref{peq-thm4-00a}.

  Next, we show~\eqref{peq-thm4-00b}. Define $\delta A = w
  e_je_i^{\top}$, $\mathcal{Y}=\left(I-\vert A\vert-\vert w\vert
    e_je_i^{\top}\right)^{-1}$, and $\Hc_{\Yc}=\Yc \BB\Yc\T$.  From
  \eqref{peq-thm4-00a},
  \begin{align*}
    \tr(\Wc_{ A+\delta A})\leq\tr(\mathcal{W}_{A+\delta A}^{\infty})
    \leq \tr(\Hc_{\Yc}).
  \end{align*}
  In addition, from \cite[Section 0.7.4]{RAH-CRJ:12}, we have
  \begin{align*}
    \mathcal{Y}&=(I-\vert A\vert)^{-1}+\frac{\vert w\vert (I-\vert
      A\vert)^{-1}e_je_i^{\top}(I-\vert A\vert)^{-1}}{1-\vert w\vert
      e_i^{\top}(I-\vert A\vert)^{-1}e_j}
  \end{align*}
  Therefore, $\mathcal{Y} 
  =\mathcal{X}+\alpha_{ij}\mathcal{X}e_je_i^{\top}\mathcal{X}$, and
  hence,
  \begin{align}
    \label{peq-thm4-3}
    \notag\Hc_{\Yc}&= \Xc\BB\Xc\T+\alpha_{ij}\Xc\BB\Xc\T
    e_ie_j\T\Xc\T
    \\
    & \quad + \alpha_{ij}\Xc e_je_i\T\Xc\BB\Xc\T
    \\
    &\quad +\alpha_{ij}^2\Xc e_je_i\T\Xc\BB\Xc\T e_ie_j\T\Xc\T. \notag
  \end{align}
  Taking trace throughout, using $\tr(U_1U_2)=\tr(U_2U_1)$ repeatedly,
  and the fact that  $\alpha_{ij} \le \alpha$ from \eqref{defs-1},
  \begin{align}
    \label{peq-thm4-4}
    \notag\tr(\Hc_{\Yc}) & \leq \tr(\Xc\BB\Xc\T)
    \\
    &\quad +\alpha\tr\Big(\Xc\BB\Xc\T( e_ie_j\T\Xc\T+\Xc
    e_je_i\T)\Big)
    \\
    & \quad +\alpha^2\tr(e_j\T\Xc\T\Xc e_je_i\T\Xc\BB\Xc\T
    e_i). \notag
  \end{align}
  Using~\cite[Proposition 8.4.13]{DSB:09}, $\tr\Big(\Xc\BB\Xc\T(
  e_ie_j\T\Xc\T+\Xc e_je_i\T)\Big)\leq \lambda_1( e_ie_j\T\Xc\T+\Xc
  e_je_i\T)\tr(\Xc\BB\Xc\T)=\lambda_1( e_ie_j\T\Xc\T+\Xc
  e_je_i\T)\tr(\Hc_{\Xc})$.
  From Lemma \ref{lem4}, $\lambda_1( e_ie_j\T\Xc\T+\Xc e_je_i\T)\leq
  \max\{2e_i\T\Xc e_j,e_i\T\Xc e_j+\Vert \Xc e_j\Vert\}$. Noting that
  $ \tr(e_j\T\Xc\T\Xc e_je_i\T\Xc\BB\Xc\T e_i) = e_j\T\Xc\T\Xc
  e_je_i\T\Xc\BB\Xc\T e_i$, we get from~\eqref{defs-1} that
  \begin{align}
    \label{peq-thm4-5}
    \notag \tr\Big(\Xc\BB\Xc\T( e_ie_j\T\Xc\T+\Xc e_je_i\T)\Big) &
    \leq \beta\tr(\Hc_{\Xc}),
    \\
    \tr(e_j\T\Xc\T\Xc e_je_i\T\Xc\BB\Xc\T e_i)&\leq
    \gamma\bar{\gamma},
  \end{align}
  which together with \eqref{peq-thm4-4} yields the result.
\end{proof}

\begin{remark}\longthmtitle{Comparison of upper bound of trace of
    Gramian with the literature}\label{rem-trace}
  {\rm The work~\cite{TM-NF-MK:82} derives a different upper bound for
    the trace of Gramian as
    \begin{align}\label{eq:literature}
      \tr(\Wc_A) \le \frac{\tr(BB^{\top})}{1-\lambda_1(AA^{\top})} ,
    \end{align}
    which is valid under the assumption that
    $\sqrt{\lambda_1(AA^{\top})}<1$. This restricts its utility as
    this condition might be violated even if $\rho(\vert A\vert)<1$.
    In contrast, the upper bound~\eqref{peq-thm4-00a} in
    Theorem~\ref{thm4} needs only $\rho(\vert A\vert)<1$ and hence is
    more generally applicable.  The $1000$ random Erd\H{o}s-R\'{e}nyi
    networks (used in our numerical simulations, cf.
    Section~\ref{ex-2}), all have $\sqrt{\lambda_1(AA^{\top})}<1$ and
    $\rho(\vert A\vert)<1$, and we observe that the upper
    bound~\eqref{peq-thm4-00a} is tighter than the upper
    bound~\eqref{eq:literature}.}\oprocend
\end{remark}

Theorem~\ref{thm4} can be used to lower bound the trace of the inverse
of the Gramian. The following result follows by using $\tr(U^{-1})\geq
\frac{n^2}{\tr(U)}$ , for $U\succeq 0$, cf.~\cite[Fact
8.12.2]{DSB:09}.

\begin{theorem}\longthmtitle{Lower bound on trace inverse of
    Gramian}\label{trace-inverse-lower-bound}
  Given a network with adjacency matrix $A$ and $\rho(\vert
  A\vert)<1$, consider the modified network resulting from adding the
  weight $w \in \real$ to its edge $i\longrightarrow j$, $i\neq j$,
  such that $\rho(\vert A\vert +\vert w\vert e_je_i\T)<1$. Then,
  \begin{align*}
    \tr(\Wcinv_A) & \geq \frac{n^2}{\tr(\Hc_{\Xc})},
    \\
    \tr({\Wcinv_{A+we_je_i\T}}) & \geq
    \frac{n^2}{(1+\alpha\beta)\tr(\Hc_{\Xc}) +
      \alpha^2\gamma\bar{\gamma}}.
  \end{align*}
\end{theorem}

%
%
%
%

Next, we derive an upper bound for the log-det of Gramian.

\begin{theorem}\longthmtitle{Upper bound on log-det of
    Gramian}\label{thm6}
  Given a network with adjacency matrix $A$ and $\rho(\vert
  A\vert)<1$, consider the modified network resulting from adding the
  weight $w \in \real$ to an arbitrary edge $i\longrightarrow j$,
  $i\neq j$, such that $\rho(\vert A\vert +\vert w\vert
  e_je_i\T)<1$. Then,
  \begin{subequations}
    \begin{align}\label{peq-thm6-0a}
      &\log\det(\Wc_A) \leq \sigma
      n\log\Big(\frac{\tr(\Wc_A)}{n^{{1}/{\sigma}}}\Big) \leq
      \sigma
      n\log\Big(\frac{\tr(\Hc_{\Xc})}{n^{{1}/{\sigma}}}\Big),
    \end{align}
    with $\sigma=1$ if $\tr(\Wc_A)\leq 1$ and $\sigma=2$ if
    $1<\tr(\Wc_A)$ and
    \begin{align}\label{peq-thm6-0b}
      &\log\det(\Wc_{A+w e_je_i\T}) \leq \sigma
      n\log\Big(\frac{\tau}{n^{{1}/{\sigma}}}\Big),
    \end{align}
    with $\tau :=
    (1+\alpha\beta)\tr(\Hc_{\Xc})+\alpha^2\gamma\bar{\gamma}$ and $
    \sigma=1$ if $\tau\leq 1$ and $\sigma=2$ for $ 1<\tau$.
  \end{subequations}
\end{theorem}
\begin{proof}
  Let $Z$ be an arbitrary adjacency matrix with Gramian $\Wc_Z$. We
  use~\cite[Theorem 2]{ZB-GHG:96} to write $
  \log\det\left(\Wc_Z\right)\leq \phi_1(\phi_2+\phi_3)$, where
  \begin{align*}
    \phi_1 & = \frac{1}{\lo\wo^2-\lo^2\wo}, \quad \phi_2 =
    (\mu_1\wo^2-\mu_2\wo) \log(\lo),
    \\
    \phi_3 & =(\mu_2\lo-\mu_1\lo^2) \log(\wo),
    \\
    \wo & = \frac{\lo \mu_1-\mu_2}{\lo n-\mu_1},
  \end{align*}
  for $\mu_1=\tr(\Wc_Z)$, $\lambda_1(\Wc_Z)\leq \lo\leq \mu_1$ and $
  \mu_2=\Vert \Wc_Z \Vert^2$.  Direct substitution of the value of
  $\wo$ and some simplifications lead~to
  \begin{align*}
    \phi_1\phi_2 & =
    \frac{n\mu_2-\mu_1^2}{\mu_2+n\lo^2-2\mu_1\lo}\log(\lo),
    \\
    \phi_1\phi_3 & = \frac{(n\lo-\mu_1)^2}{\mu_2+n\lo^2-2\lo\mu_1}\log
    \bigg(\frac{\lo\mu_1-\mu_2}{n\lo-\mu_1}\bigg).
  \end{align*}
  We look at these as functions of~$\mu_2$. From Lemma~\ref{lem5},
  $\mu_2\in \Big[\frac{\mu_1^2}{n},\;\lo\mu_1\Big]$.  Since
  $\phi_1\phi_3$ is a decreasing function of $\mu_2$, its maximum
  value occurs at $\mu_2=\frac{\mu_1^2}{n}$, yielding $
  n\log(\mu_1)-n\log(n)$.  To find the maximum value of
  $\phi_1\phi_2$, we compute its derivative as $\frac{d}{d\mu_2}
  (\phi_1\phi_2) =
  \Big(\frac{n\lo-\mu_1}{\mu_2+n\lo^2-2\mu_1\lo}\Big)^2\log(\lo)$. If
  $\log(\lo)\leq 0$, then $\phi_1\phi_2$ is decreasing and has maximum
  value of $0$ at $\mu_2=\frac{\mu_1^2}{n}$. If $\log(\lo)> 0$, then
  $\phi_1\phi_2$ is increasing and has maximum value of
  $\frac{\mu_1}{\lo}\log(\lo)$ at $\mu_2=\lo\mu_1$.  From the above
  discussion, and from $\lo\leq \mu_1 \leq n\lo$, we deduce
  \begin{align}
    \label{peq-thm6-1} 
    \log\det(\Wc_Z)\leq 
    \begin{cases}
      n\log(\mu_1)-n\log(n) & \text{if}\; \mu_1\leq 1,
      \\
      2n\log(\mu_1)-n\log(n) &  \text{if}\; \mu_1> 1.
    \end{cases}
  \end{align}
  For $Z=A$, \eqref{peq-thm6-0a} follows from~\eqref{peq-thm6-1} by
  using $\tr(\Wc_A) \leq \tr(\Hc_{\Xc})$, cf. Theorem~\ref{thm4}.  For
  $Z=A+we_je_i^{\top}$, \eqref{peq-thm6-0b} follows
  from~\eqref{peq-thm6-1} by using $\tr(\Wc_{A+w e_je_i\T})\leq
  (1+\alpha\beta)\tr(\Hc_{\Xc})+\alpha^2\gamma\bar{\gamma}$,
  cf. Theorem~\ref{thm4}.
%
\end{proof}

Note that the upper bound on the log-det of the Gramian in
Theorem~\ref{thm6} depends on the trace of Gramian and the size of the
network.

\subsection{Ensuring Stability of the Modified   Network}\label{stability-modified-network-1}

The modification of edges in a network can lead to instability. Here
we provide bounds on the edge-weight addition parameter $w \in \real$
that ensures instability does not arise.  The next result
generalizes~\cite[Theorem 1]{GL-CA:20-ifac} by considering arbitrary
edge weight modifications instead of only positive ones.

\begin{theorem}\longthmtitle{Bounds on edge weight perturbation
    resulting in stable modified network}\label{thm:edge-weight-bounds-1}
  Consider the network dynamics \eqref{peq-2.0.1} with a stable
  adjacency matrix $A$ and let $w\in \real$ be the edge modification
  on the edge $i \longrightarrow j$.  If $\rho(\vert A\vert)<1$ and
  $\frac{-1}{\mathcal{X}_{ij}}<w<\frac{1}{\mathcal{X}_{ij}}$, then the
  modified network $A+we_je_i^{\top}$ is stable.
\end{theorem}
\begin{proof}
  We seek to prove that $\rho(A+we_je_i^{\top})<1$.  We know that
  $A\leq \vert A\vert$ and $A+we_je_i^{\top}\leq \vert A \vert +\vert
  w \vert e_je_i^{\top}$.  Consider $(I-\vert A\vert-\vert w\vert
  e_je_i^{\top})^{-1}$.  Using~\cite[Section 0.7.4]{RAH-CRJ:12},
  \begin{align}
    \label{peq-thm3.4.1}
    (I-\vert A\vert-\vert w\vert
    e_je_i^{\top})^{-1}=\mathcal{X}+\frac{\vert w\vert
      \mathcal{X}e_je_i^{\top}\mathcal{X}}{1-\vert w\vert
      e_i^{\top}\mathcal{X}e_j}.
  \end{align}
  As $\rho(\vert A\vert)<1$, we have $\mathcal{X}\geq 0$ from
  \cite[Theorem 1.2]{JD-AZ:09}.  So, if $1-\vert w\vert
  e_i^{\top}(I-\vert A\vert)^{-1}e_j>0$, then $(I-\vert A\vert-\vert
  w\vert e_je_i^{\top})^{-1}\geq 0$. Note that
  $\frac{-1}{\mathcal{X}_{ij}}<w<\frac{1}{\mathcal{X}_{ij}}$ is
  equivalent to the condition $1-\vert w\vert e_i^{\top}(I-\vert
  A\vert)^{-1}e_j>0$. Therefore, in such case, and using \cite[Theorem
  1.2]{JD-AZ:09}, we get $\rho(\vert A \vert +\vert w \vert
  e_je_i^{\top})<1$.  From \cite[Theorem 8.1.18]{RAH-CRJ:12},
  $\rho(A+we_je_i^{\top})\leq \rho(\vert A \vert +\vert w \vert
  e_je_i^{\top})$ and the required result follows.
\end{proof}

Note that Theorem~\ref{thm:edge-weight-bounds-1} is a sufficient
condition for general networks, i.e., for $A\in\real^{n\times n}$ and
$w\in\real$, but is a necessary and sufficient condition for $A\geq0$
and $w\ge 0$, cf.~\cite[Theorem 1]{GL-CA:20-ifac}.  The element
$(j,i)$ of $\vert A \vert^t$ represents the sum of the products of the
(absolute values of) weights of all paths from node $i$ to node $j$ of
length exactly $t$, where $t$ is an arbitrary positive integer. Hence,
from~\eqref{peq-lem3-1}, if there is a path (of any length) from node
$i$ to node $j$, then the $(j,i)$ element of $\Xc$ is non-zero.  From
Theorem~\ref{thm:edge-weight-bounds-1}, the bounds on the modification
weight $w$ of edge $i\longrightarrow j$ depend upon the $(i,j)$
element of $\Xc$, i.e., on the existence of a path from node $j$ to
node $i$. Therefore, a stable network could be made unstable by
forming cycles in the network through edge modifications. These
observations lead us to state the following result regarding the
stability of stem-bud networks after edge modification.

\begin{corollary}\longthmtitle{Stability of modified stem-bud
    networks}\label{cor:stem-bud-stability-1}
  A stable stem-bud network with $A\geq 0$ remains stable if
  modifications $w\ge 0$ are performed to the  weight of edges in its
  stem.
\end{corollary}

Consequently, a directed line network can be made unstable only by
adding a new edge (that will create a bud and hence break its pure
line structure).  A directed ring network can be made unstable by
suitably modifying the weight of any of its edges. If all edges of a
stem-bud network have equal positive weights and only edges of the bud
are modified, then the bounds on the edge weight modification and
upper bound on the trace of Gramian depend upon the length of the bud,
as shown next.
  
\begin{theorem}\longthmtitle{Stem-bud networks with
    equal edge weights}\label{stem-bud-analysis}
  Consider a stable stem-bud network with $A\geq 0$, all edge weights
  equal to $0\leq a<1$, and a single input at node 1.  We have the following,
  \begin{enumerate}
  \item Let $w>0$ denote the weight modification performed only on the
    edges of the bud.  Then, the modified network is stable if $0< w
    <\frac{1-a^{L_b}}{a^{L_b-1}}$;
  \item $\tr(\Wc_A)\leq \sum_{k=0}^{y-2}a^{2k} +
    \sum_{k=y-1}^{n-1}\frac{a^{2k}}{(1-a^{L_b})^2}$.
  \end{enumerate}
\end{theorem}
\begin{proof}
  To show the result, we rely on
  Theorems~\ref{thm:edge-weight-bounds-1} and~\ref{thm4}. Before
  invoking them, we compute the expression for the elements
  of~$\Xc=(\Xc_{ij})$. Since $A\geq0$ and $\rho(A)<1$ with $0\leq
  a<1$, we have from \eqref{peq-lem3-1} that $\Xc =
  \sum_{t=0}^{\infty}A^t$. Note that the ${ij}$th element of $A^t$
  corresponds to the product of weights in the path of length $t$ (if
  it exists) from node $j$ to node~$i$.  Any node $i$ in the stem
  ($1\leq i <y$) can only be reached by the input once while any node
  $i$ in the bud ($y\leq i \leq n$) can be reached multiple
  times. Therefore, for $1\leq i<y$, we have
  \begin{align*}
    \Xc_{ij} &=
    \begin{cases}
      0 & \text{for } i<j,
      \\
      a^{i-j} & \text{for } i\geq j,
    \end{cases}
  \end{align*}
  and for $y\leq i\leq n$, we have
  \begin{align*}
    \Xc_{ij} &= 
    \begin{cases}
      \sum_{k=0}^{\infty}
      a^{L_b+i-j}a^{kL_b}=\frac{a^{L_b+i-j}}{1-a^{L_b}} & \text{for
      } y\leq i<j,
      \\
      \sum_{k=0}^{\infty} a^{i-j}a^{kL_b} =\frac{a^{i-j}}{1-a^{L_b}}
      & \text{for } j\leq i\leq n,
    \end{cases}
  \end{align*}
  where we have used the infinite geometric series formula
  $\sum_{k=0}^{\infty}a_1r^k = \frac{a_1}{1-r}$, for $\vert r \vert
  <1$.

  To establish (i), we resort to
  Theorem~\ref{thm:edge-weight-bounds-1}.  Since only edges in the bud
  are modified, to ensure stability it is enough to enforce that $|w|$
  is smaller than $1/\Xc_{i-1,i}$, for $y+1\leq i \leq n$, and
  $1/\Xc_{ny}$. The minimum of these values is
  $\frac{1-a^{L_b}}{a^{L_b-1}}$, and the result follows.  Regarding
  (ii), since there is only one input at node~$1$, $B = e_1$, and
  hence $\Xc B$ corresponds to the first column of~$\Xc$. Therefore, $
  \tr(\Hc_{\Xc}) =
  \sum_{k=0}^{y-2}a^{2k}+\sum_{k=y-1}^{n-1}\frac{a^{2k}}{(1-a^{L_b})^2}$,
  and the result follows from Theorem~\ref{thm4}.
\end{proof}


\section{Numerical Examples}\label{examples-1}
Here, we provide two sets of simulations to illustrate the efficacy of
the proposed edge centrality measure and the bounds derived for the
performance metrics.  All the simulations are performed using MATLAB
on a desktop with Intel core-i7-8700, $3.20$ GHz processor with $16$
GB of RAM. We consider the objective function~$f$ as the trace and as the
log-det of the Gramian and use the following notation in both
examples: $f_I$ denotes the value of $f $ for the original network;
$f_{EC}$ denotes the value of $f$ for the network after modification
of the edge with the best edge centrality; and $f_{EX}$ denotes the
global maximum of $f$ computed by exhaustive search of all single-edge
modifications.


\subsection{6-Node Stem-Bud Networks}\label{ex-1}
Consider a family of $6$-node stem-bud networks, all with parameters
$a_{21}=0.9$, $a_{32}=0.7$, $a_{43}=0.8$, $a_{54}=0.6$, and
$a_{65}=0.8$. We consider different networks depending on where the
junction node $y$ lies, i.e., $0\leq y \leq 5$. Figure~\ref{6-node-1}
shows an example with $y=2$. When $y=0$, the network is a directed
line and when $y=1$, the network is a directed ring.  Whenever a
backward edge is present, we set $a_{y6}=0.7$.  We consider actuators
at nodes $1$ and~$3$, and take $T=2n$. In all cases, the
controllability Gramian is a diagonal matrix, validating
Proposition~\ref{prop1}.

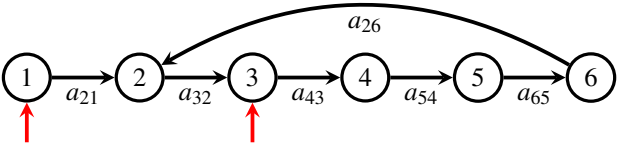
\begin{figure}[htb]
  \centering
  \begin{tikzpicture}[>=stealth,line width=0.43mm]
    \node[shape=circle,draw=black,line width=0.43mm] (no1) at
    (0,0){$1$}; \node[shape=circle,draw=black,line width=0.43mm] (no2)
    at (1.5,0){$2$}; \node[shape=circle,draw=black,line width=0.43mm]
    (no3) at (3,0){$3$}; \node[shape=circle,draw=black,line
    width=0.43mm] (no4) at (4.5,0){$4$};
    \node[shape=circle,draw=black,line width=0.43mm] (no5) at
    (6,0){$5$}; \node[shape=circle,draw=black,line width=0.43mm] (no6)
    at (7.5,0){$6$};
    \draw[->,line width = 0.5mm](no1)--node[below]{$a_{21}$}(no2);
    \draw[->,line width = 0.5mm](no2)--node[below]{$a_{32}$}(no3);
    \draw[->,line width = 0.5mm](no3)--node[below]{$a_{43}$}(no4);
    \draw[->,line width = 0.5mm](no4)--node[below]{$a_{54}$}(no5);
    \draw[->,line width = 0.5mm](no5)--node[below]{$a_{65}$}(no6);
    \draw[->,line width=0.5mm](no6) edge[bend right] node [below] {$a_{26}$}(no2);
    
    \node[] (i1) at (0,-1){};
    \node[] (i2) at (3,-1){};
    
    \draw[->,red,line width=0.5mm](i1) -- (no1);
    \draw[->,red,line width=0.5mm](i2) -- (no3);
  \end{tikzpicture}
  \caption{Example of a $6$-node stem-bud network with inputs at node
    $1$ and node $3$ with junction at node $2$. The structure of ECM
    for this particular network is shown in Figure
    \ref{pat2}.}\label{6-node-1}
\end{figure}

Following Theorem~\ref{thm31}, Table~\ref{table-11} describes the
properties of the edge centrality matrix as the junction node $y$ goes
from $0$ to $5$, including the list of sub- and super-diagonals having
non-zero elements. As $y$ moves towards the end node $n$, the bud
length decreases, but more edges (existing as well as non-existing)
become influential.  Note that the structure of the ECMs is the same
irrespective of the number of inputs and the performance
metric. Figure~\ref{fig:sparsity-1} shows the corresponding sparsity
patterns of ECMs for each network.

\begin{table}[htb]
  \begin{center}
    \setlength{\extrarowheight}{1.6pt}
    \begin{tabular}{|c|c|c|c|c|c|c|c|}
      \hline
      \multicolumn{2}{|c|}{Junction node $y$} & $0$ & $1$ & $2$  &
      $3$ & $4$ & $5$
      \\
      \hline
      \multicolumn{2}{|c|}{Bud length $L_b$} & $\infty$  & $6$  & $5$ & $4$ & $3$ &
      $2$
      \\
      \hline
      \multirow{2}{*}{Sub-diagonals} & $k_{\sub}$ & $0$ & $0$ & $0$ &
      $1$ & $1$ & $2$
      \\ \cline{2-8}
      & $N_{\sub}$ & $\{1\}$ & $\{1\}$ & $\{1\}$ & $\! \{1,5\}\!$  &
      $\! \{1,4\} \!$  & $\! \{1,3,5\}\!$ 
      \\ 
      \hline	
      \multirow{2}{*}{Super-diagonals} & $k_{\sup}$ & $0$ & $1$ & $1$
      & $1$ & $2$ & $3$\\
      \cline{2-8}
      & $N_{\sup}$ & $\emptyset$ & $\{5\}$ & $\{4 \}$ & $\{3\}$  &
      $\! \{2,5\} \! $
      & $\! \{1,3,5\} \!$ 
      \\ 
      \hline      	    	     	    	
    \end{tabular} 
  \end{center}
  \caption{Structure of  ECM for a $6$-node stem-bud network as a
    function of the  junction node location. Columns correspond, in
    order, to plots (a)-(f) in Figure~\ref{fig:sparsity-1}. }\label{table-11}
\end{table}

\begin{figure}[htb]
  \centering
  \subfloat[$y=0$]{\label{pat0}\includegraphics[width=0.33\columnwidth]{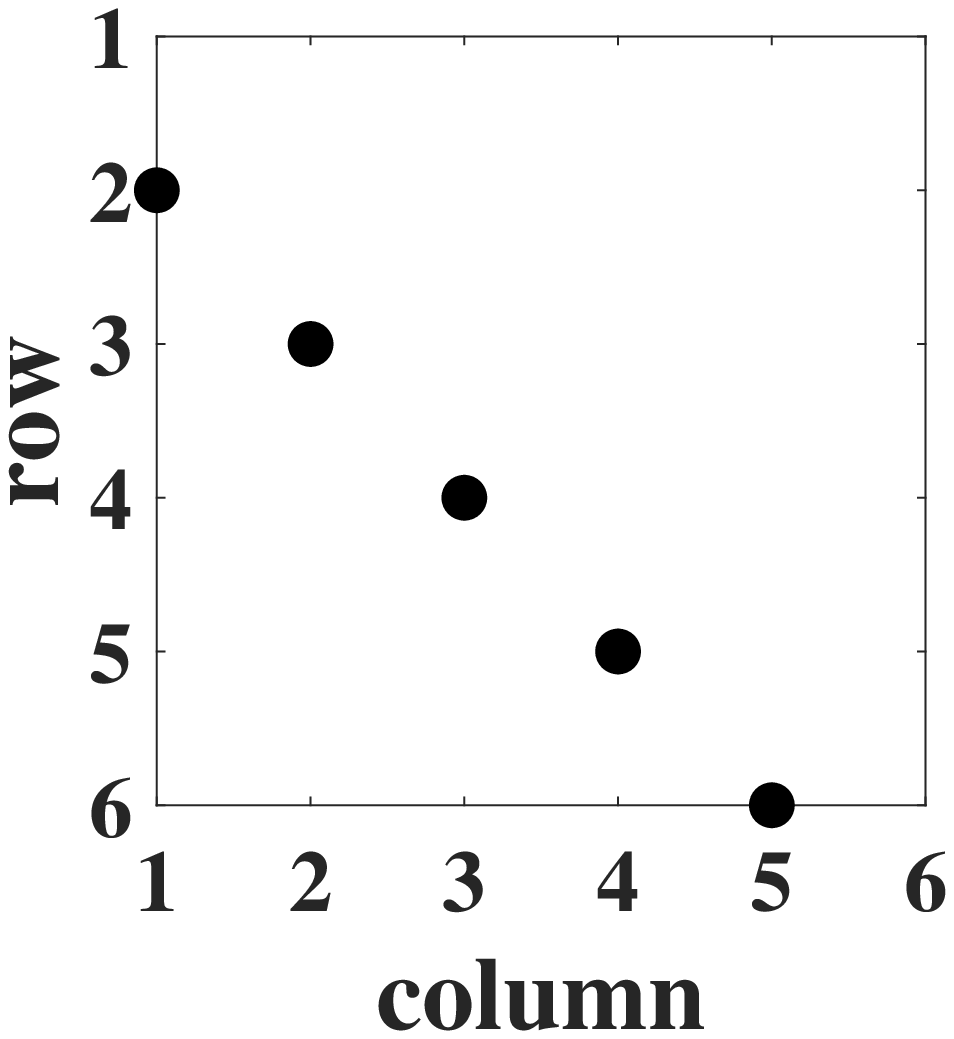}
  }
  \subfloat[$y=1$]{\label{pat1}\includegraphics[width=0.33\columnwidth]{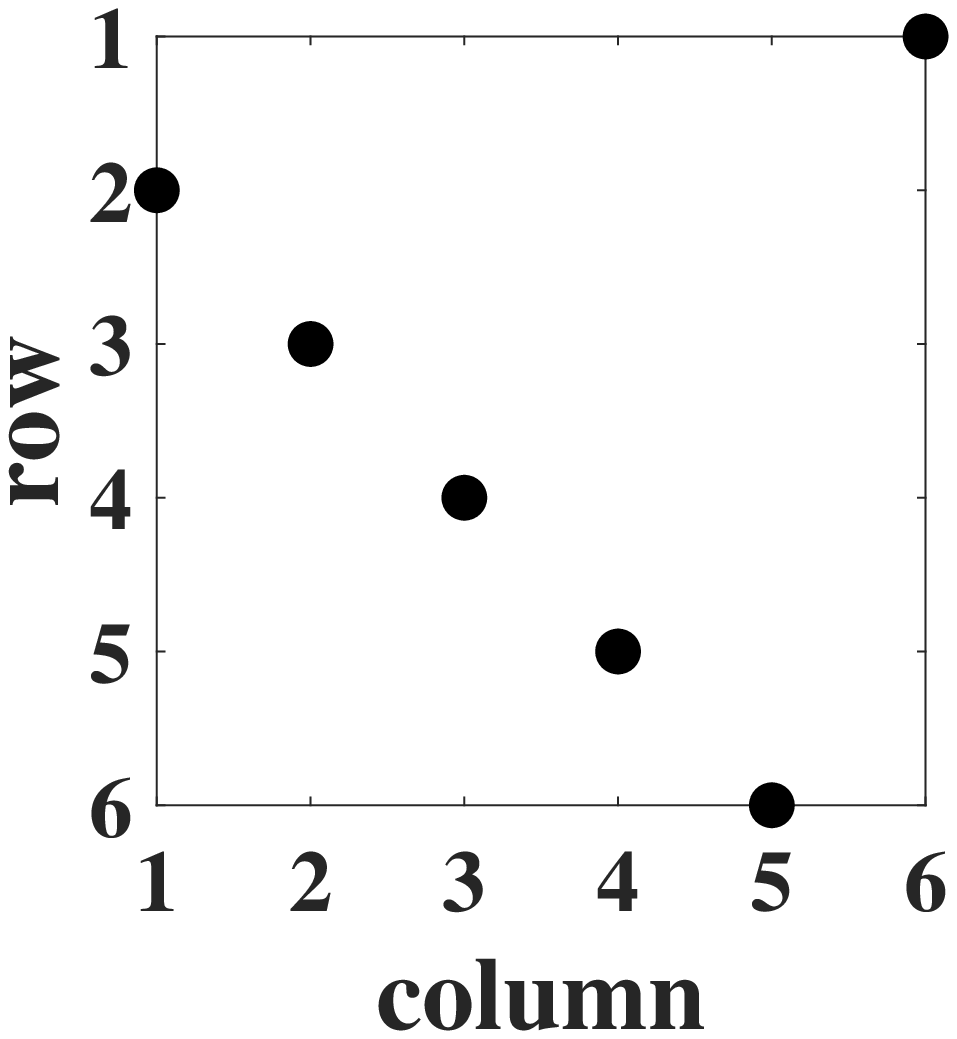}}
  \subfloat[$y=2$]{\label{pat2}\includegraphics[width=0.33\columnwidth]{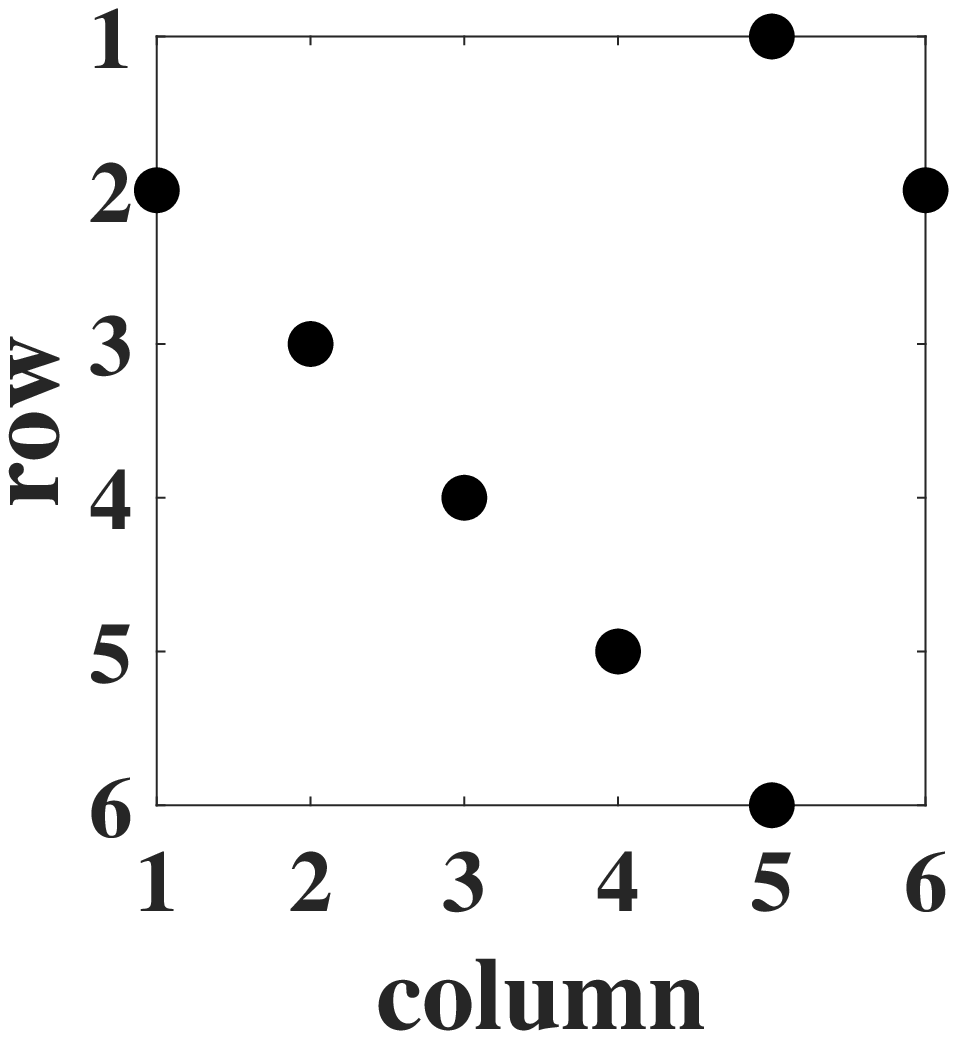}
  }
  \\
  \subfloat[$y=3$]{\label{pat3}\includegraphics[width=0.33\columnwidth]{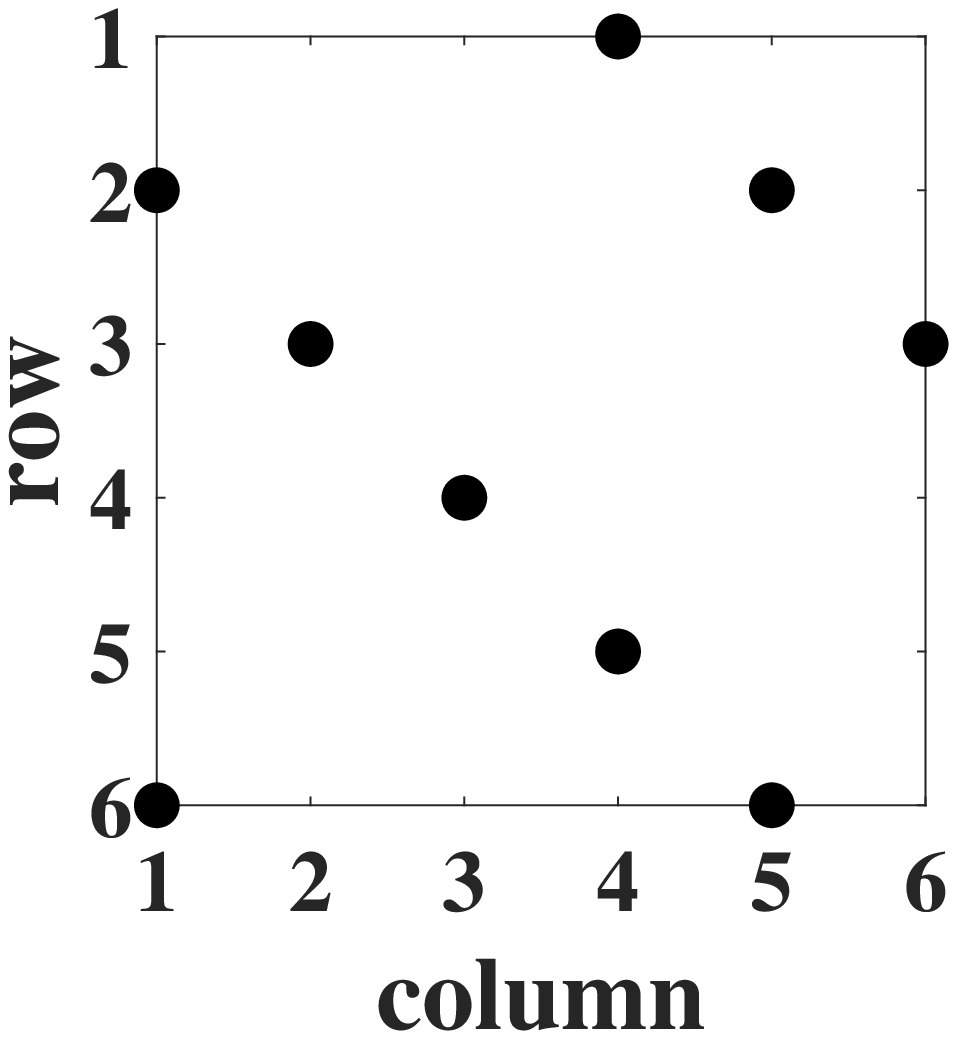}}
  \subfloat[$y=4$]{\label{pat4}\includegraphics[width=0.33\columnwidth]{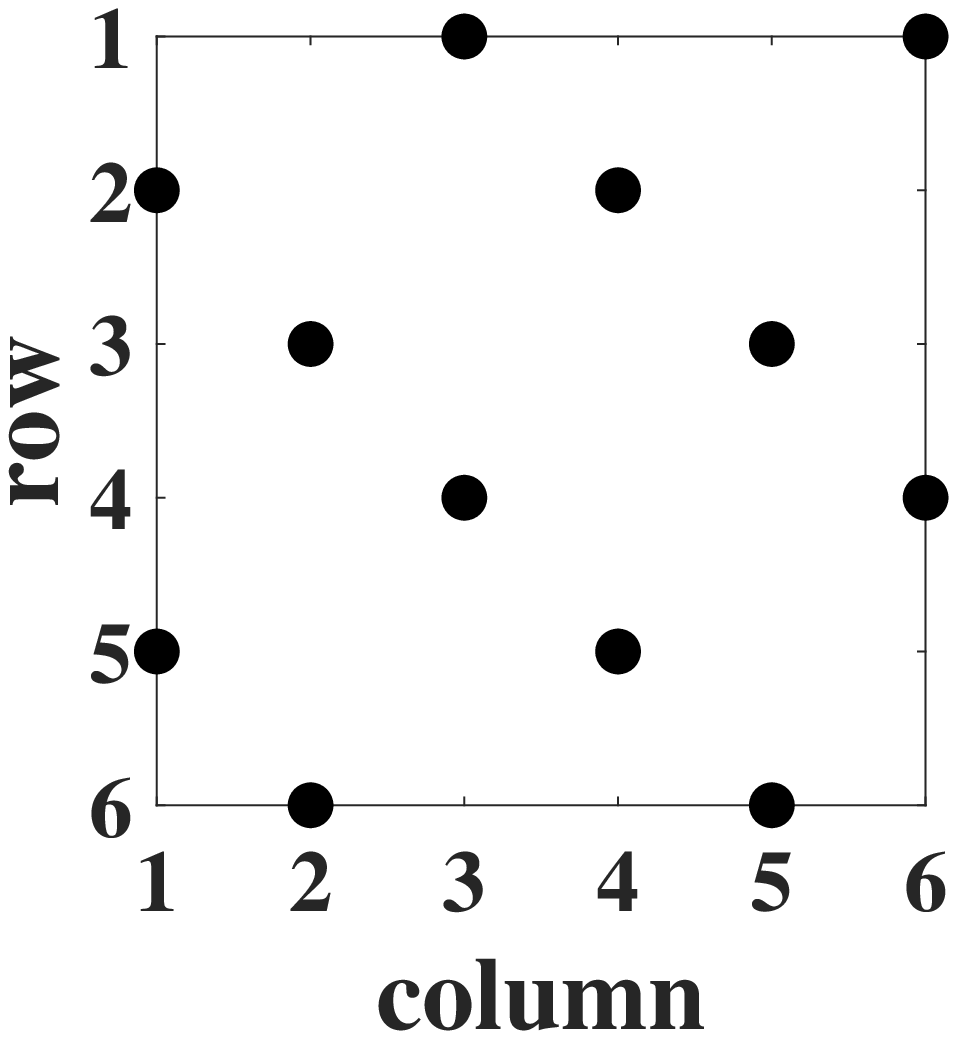}
  }
  \subfloat[$y=5$]{\label{pat5}\includegraphics[width=0.33\columnwidth]{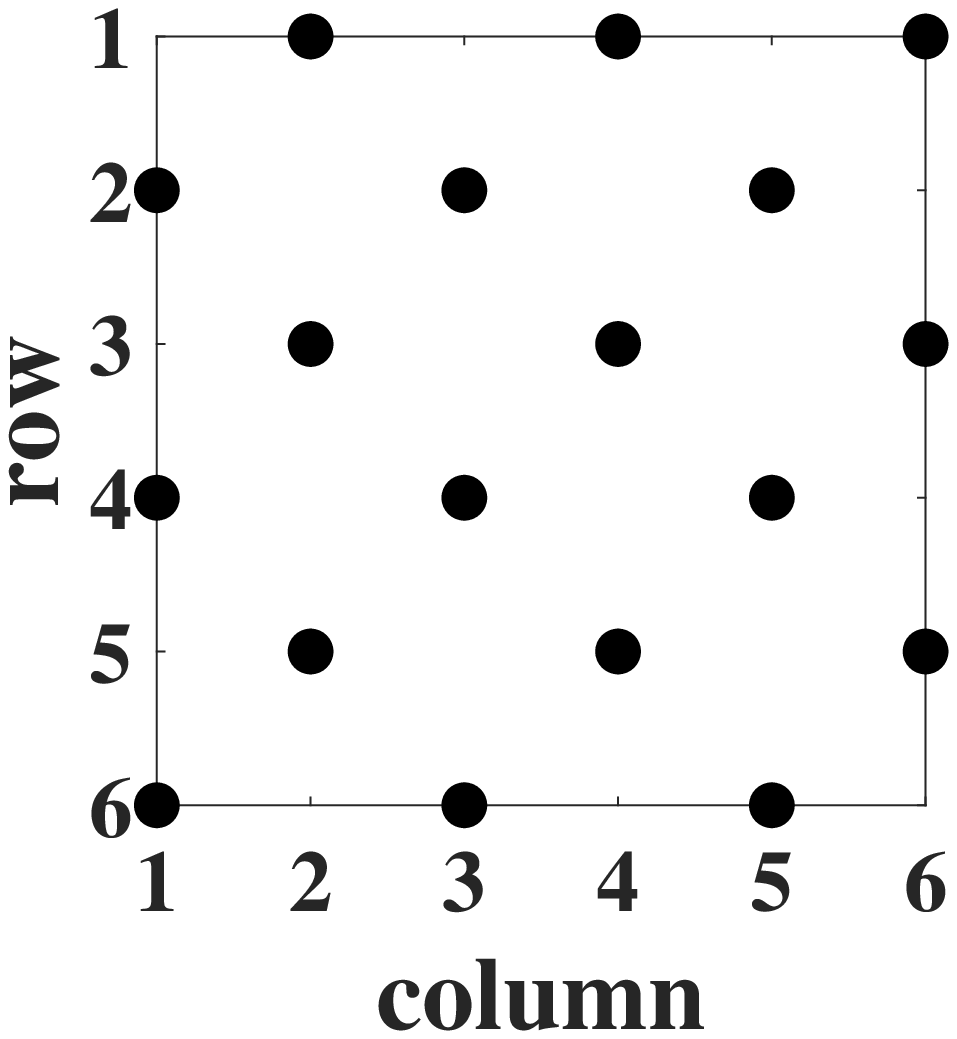}}
  \caption{Sparsity pattern of ECM for a $6$-node stem-bud network as
    a function of junction node location with either trace or log-det
    as performance metric.  `$\bullet$' represents a non-zero
    element.}\label{fig:sparsity-1}%
\end{figure}

We analyze the improvements in network controllability with a single
edge modification (excluding self-loops).  For each possible junction
node $y$, we first compute $\Xc$ and determine its largest
non-diagonal element. According to
Theorem~\ref{thm:edge-weight-bounds-1}, this determines the maximum
allowable perturbation weight, denoted $w_{\max}$, which ensures
stability of the modified network. We set $w=0.99w_{\max}$ as the edge
modification weight and exhaustively search for the global solution,
including the edges not present in the network. Table~\ref{table-12}
shows the results of the comparison with the solution obtained by
modifying the edge with the best edge centrality. We observe that, for
the trace of the Gramian, the global solution and the best edge
centrality solution are the same except for the directed line and when
the junction node is at $3$ or $4$. For the log-det case, the best
edge centrality solution matches the global solution for all the
junction nodes.  This shows the usefulness of the proposed edge
centrality notion even when the modification weight is large.

\begin{table}[htb]
  \begin{center}
    \setlength{\extrarowheight}{1.6pt}
    \begin{tabular}{|c|c|c|c|c|c|c|c|}
      \hline
      \multicolumn{2}{|c|}{\backslashbox{Metric}{\vspace{-5pt} $y$}} & $0$ & $1$ & $2$  & $3$ & $4$& $5$ \\
      \hline   
      \multicolumn{2}{|c|}{$w$} & $1.1$  & $0.91$ & $0.89$& $0.90$& $0.82$& $0.55$	 \\
      \hline
      \multirow{3}{*}{Trace} & $f_I$ & $4.63$ & $4.90$&$4.83$ & $4.84$ & $4.81$  & $4.87$  \\ \cline{2-8}
      & $f_{EC}$ & $11$ & $11.7$ & $11.4$ & $13$ & $9.8$ & $7.9$ \\ 
      \cline{2-8}
      & $f_{EX}$ & $16$ & $11.7$ & $11.4$ & $14.3$ & $11.6$ & $7.9$ \\
      \hline    
      \multirow{3}{*}{Log-det} & $f_I$ & $-2.7$ & $-2.39$&$-2.44$ & $-2.4$ & $-2.39$  & $-1.96$  \\ \cline{2-8}
      & $f_{EC}$ & $2.5$ & $3.2$ & $2.9$ & $3.2$ & $3$ & $1.2$ \\ 
      \cline{2-8}
      & $f_{EX}$ & $2.5$ & $3.2$ & $2.9$ & $3.2$ & $3$ & $1.2$ \\ 
      \hline	     	    	
    \end{tabular} 
  \end{center}
  \caption{Improvement  in performance metrics with single-edge modification for
    a $6$-node stem-bud network as a function of  junction node
    location. $w$ denotes the edge modification weight.}\label{table-12}
\end{table}

\subsection{Random Erd\H{o}s-R\'{e}nyi Networks}\label{ex-2}
Here, we show the efficacy of ECM and the utility of the performance
metric bounds obtained in Section~\ref{global-1} on $1000$ random
Erd\H{o}s-R\'{e}nyi (ER) networks~\cite{PE-AR:60} without self-loops.
We consider networks with $n = 100$ nodes, $m=30$ input nodes, and
edge probability of $0.35$. For uniformity, the spectral radius
$\rho(A)$ of each ER network belongs to the interval $(0.85,0.90)$.
The time horizon is $T=n$ and we neglect self-loops.  Using
Theorem~\ref{thm:edge-weight-bounds-1}, we compute $3.82$ as the bound
on the edge modification weight beyond which stability of the random
ER networks cannot be guaranteed.  We compute the improvement $f_{EC}$
in the performance metrics obtained by selecting the edge with the
best centrality, and compare this improvement with that obtained
using exhaustive search, $f_{EX}$.

\begin{table}[!htb]
  \begin{center}
    \setlength{\extrarowheight}{1.6pt}
    \begin{tabular}{|c|c|c|c|c|c|c|c|c|}
      \hline
      \multicolumn{2}{|c|}{\backslashbox{Metric}{\vspace{-5pt}$w$}} &
      $0.5$ & $1$ & $1.5$  & $2$ & $2.5$ & $3$ & $3.5$
      \\ 
      \hline
      \multicolumn{2}{|c|}{Avg $\vert f_I \vert$} & $32$  & $32$  &
      $32$ & $32$ & $32$ & $32$ & $32$ 
      \\
      \hline
      \multirow{3}{*}{{$\%$ $\frac{f_{XC}}{\vert f_I\vert}$}} & Worst
      & $0.06$ & $0.32$ & $0.92$ & $2.16$ & $4.80$ & $10.4$ &$23.8$\\
      \cline{2-9} 
      & Best & $0$ & $0$ & $0$ & $0$  & $0$  & $0$  & $0$ \\ \cline{2-9}
      & Avg & $0.00$ & $0.04$ & $0.18$ & $0.61$ & $1.63$ & $3.91$
      &$9.17$ 
      \\
      \hline  
      \multicolumn{2}{|c|}{Avg $f_{XC}$} & $0.00$  & $0.01$  & $0.06$
      & $0.20$ & $0.51$ & $1.25$ & $2.93$ 
      \\
      \hline
      \multirow{3}{*}{{$\%$ $\frac{f_{CI}}{\vert f_I\vert}$}} & Worst
      & $1.26$ & $4.26$ & $9.03$ & $15.6$  & $24.0$  & $34.1$  &
      $46.2$ 
      \\ \cline{2-9}
      & Best & $1.90$ & $5.97$ & $12.6$ & $22.5$ & $37.5$ & $61.5$ &
      $108$
      \\ 
      \cline{2-9}
      & Avg & $1.57$ & $5.08$ & $10.6$ & $18.4$ & $28.6$ & $41.9$
      &$59.1$ 
      \\
      \hline 
      \multicolumn{2}{|c|}{ Avg $f_{CI}$} & $0.50$  & $1.63$  & $3.39$
      & $5.90$ & $9.15$ & $13.4$ & $18.9$ 
      \\
      \hline	 
      \multicolumn{2}{|c|}{{Avg $\%$ $\frac{\vert
            f_{EX}-f_{g}\vert}{\vert f_I\vert}$}} &  $1.46$ & $3.12$
      & $5.63$ & $7.37$ & $8.91$  & $10.6$ & $14.4$ 
      \\
      \hline	      	    	     	    	
    \end{tabular} 
  \end{center}
  \caption{Performance analysis with trace of Gramian as
    objective for 1000 Random Erd\H{o}s-R\'{e}nyi networks. Here,
    $f_{XC}=\vert f_{EX}-f_{EC}\vert$ and $f_{CI}=\vert 
    f_{EC}-f_{I}\vert$.}\label{table-21}
\end{table}

\begin{table}[!htb]
  \begin{center}
    \setlength{\extrarowheight}{1.6pt}
    \begin{tabular}{|c|c|c|c|c|c|c|c|c|}
      \hline
      \multicolumn{2}{|c|}{\backslashbox{Metric}{\vspace{-5pt}$w$}} &
      $0.5$ & $1$ & $1.5$  & $2$ & $2.5$ & $3$ & $3.5$
      \\
      \hline
      \multicolumn{2}{|c|}{Avg $\vert{f_I}\vert$  } & $553$  & $553$
      & $553$   & $553$   & $553$   & $553$   & $553$     
      \\
      \hline
      \multirow{3}{*}{{$\%$ $\frac{f_{XC}}{\vert f_I\vert}$}  } & Worst & $0.56$ & $0.67$ &
      $0.72$ & $0.76$ & $0.77$ & $0.80$ & $0.83$
      \\ \cline{2-9}
      & Best &$0$ & $0$ & $0$ & $0$  & $0$  & $0$  & $0$  
      \\ 
      \cline{2-9}
      & Avg & $0.19$ & $0.24$ & $0.26$ & $0.29$ & $0.30$ & $0.32$ &
      $0.33$ 
      \\ 
      \hline 
      \multicolumn{2}{|c|}{  Avg $f_{XC}$} & $1.05$  & $1.33$  & $1.44$
      & $1.60$ & $1.66$ & $1.77$ & $1.83$ 
      \\
      \hline 
      \multirow{3}{*}{{$\%$ $\frac{f_{CI}}{\vert f_I\vert}$}  } & Worst & $1.16$ & $1.63$ &
      $1.95$ & $2.20$  & $2.43$  & $2.62$  & $2.79$ 
      \\ \cline{2-9}
      & Best & $2.59$ & $3.36$ & $3.86$ & $4.24$ & $4.54$ & $4.81$ &
      $5.03$ 
      \\ 
      \cline{2-9}
      & Avg & $1.88$ & $2.63$ & $3.12$ & $3.50$ & $3.81$ & $4.08$ &
      $4.32$ 
      \\ 
      \hline  
      \multicolumn{2}{|c|}{Avg $f_{CI}$  } & $10.4$  & $14.5$  & $17.3$
      & $19.4$ & $21.0$ & $22.6$ & $23.9$ 
      \\ 
      \hline	
      \multicolumn{2}{|c|}{ {Avg $\%$ $\frac{\vert
            f_{EX}-f_{g}\vert}{\vert f_I\vert}$}} & $1.04$  & $1.36$
      & 
      $2$ & $1.81$ & $1.49$ & $1.13$ & $0.94$
      \\ 
      \hline	    	
    \end{tabular}
  \end{center}
  \caption{Performance analysis with log-det of Gramian as
    objective  for 1000 Random Erd\H{o}s-R\'{e}nyi networks. Here,
    $f_{XC}=\vert f_{EX}-f_{EC}\vert$ and $f_{CI}=\vert 
    f_{EC}-f_{I}\vert$.} 
  \label{table-22}
\end{table}

The results are displayed in Table~\ref{table-21} for trace of Gramian
and Table~\ref{table-22} for log-det of Gramian. Regarding the
performance of the ECM-based solution versus exhaustive search for the
trace of Gramian, we observe that, as the modification weight
increases, the worst-case value for $\%\frac{f_{XC}}{\vert f_I\vert}$
increases.  Such behavior is expected as $f_{EC}$ is computed using
the ECM according to the first-order effects of edge
modification. However the average $\%\frac{f_{XC}}{\vert f_I\vert}$
remains significantly lower, less than half the worst case.  Regarding
improvement, we see that the average $\%\frac{f_{CI}}{\vert f_I\vert}$
increases with the modification weight and is more prominent for the
trace as objective with an average more than $50\%$ for $w=3.5$. In
the case of log-det of Gramian as objective, even though the values of
$f_{XC}$ and $f_{CI}$ are significant, the percentage values are small
as the initial objective value $f_I$ is large. Note that even for
higher modification weights, in some cases the use of edge centrality
leads to the global solution.
Note that the exhaustive search procedure is weight dependent, while
the ECM-based one is not. Thus, the computation of the global solution
exhaustively involves a significantly larger computational cost than
the ECM-based one as the network size increases.  For the considered
ER networks, the cost of computing the best edge using ECM is
approximately one-tenth the cost of exhaustive search for a
given weight.

We also illustrate the capability of the network performance bounds to
capture the global maxima of the corresponding metric.  For the trace
of the Gramian, we employ~\eqref{peq-thm4-00b} in Theorem~\ref{thm4}
to compute a global upper bound $h_g^{\tr}$.  Similarly, for the
log-det of the Gramian, we employ~\eqref{peq-thm6-0b} with
$\tau=h_g^{\tr}$ in Theorem~\ref{thm6} to compute the global upper
bound
\begin{align*}
  h_g^{\log\det}=2n\log \Big(\frac{h_g^{\tr}}{\sqrt{n}}\Big).
\end{align*}
We observe that these bounds of the true global optimum $f_{EX}$
(obtained through exhaustive search) tend to be conservative, so we
refine them by computing estimates $f_g^{\tr}$ and $f_g^{\log\det}$ of
the value of $f$ corresponding to the global edge modification
solution as follows.

We take the $30$ edges with the largest ECM.  For each of these edges,
we modify its weight by $w>0$ and consider the resulting adjacency
matrix $\tA$. We then compute $\Wc_{\tA}$, $\tXc=(I-\tA)^{-1}$, and
$\Hc_{\tXc}=\tXc BB^{\top} \tXc^{\top}$. According
to~\eqref{peq-thm4-00a} in Theorem~\ref{thm4}, we have
$\fb=\tr(\Wc_{\tA}) \le \hb=\tr(\Hc_{\tXc})$.  We take all the data
pairs $(\fb,\hb)$ obtained in this way and fit a curve using the
curve-fitting function \textit{`fit'} and model type \textit{`rat55'}
of MATLAB. Using the obtained curve-fit we compute the estimate
$f_g^{\tr}$ corresponding to $h_g^{\tr}$.  We run this procedure for
different modification weights $w$, from $0.5$ to $3.5$ with an
increment of $0.5$ for all considered ER networks.


We follow the same procedure for the log-det of the Gramian, computing
the appropriate bounds, now resorting to Theorem~\ref{thm6}, as
follows.  As the trace of the Gramian for all $1000$ ER networks is
greater than $1$, we use~\eqref{peq-thm6-0a} with $\sigma=2$.  For
each of the $30$ edges with the largest ECM we obtain
$\fb=\log\det(\Wc_{\tA}) \leq
\hb=2n\log\Big(\frac{\tr(\Hc_{\tXc})}{\sqrt{n}}\Big)$.  As in the case
of the trace, we take all the data pairs $(\fb,\hb)$ obtained in this
way and fit a curve, which we use to compute the estimate
$f_g^{\log\det}$ corresponding to $h_g^{\log\det}$.


The last rows of Tables~\ref{table-21} and~\ref{table-22} list the
average percentage values of the error, relative to the initial metric
value, between the global solution computed by exhaustive search and
the estimate of the global solution obtained as described above. This
error is less than $\approx 15\%$ for the trace of Gramian and less
than $\approx 2\%$ for the log-det of Gramian. Note that, for creating
the estimate of the global solution, we have used just $30$ of $9900$
edges, i.e., $0.3\%$ of the total number of edges. For the trace of
Gramian, the error between the estimate and the actual global solution
increases with the modification weight. This may be due to the
nonlinear dependence of $h^{\tr}_g$ on the modification weight~$w$.

\section{Conclusions}\label{sec:conclusions}
We have studied network systems modeled as controlled linear-time
invariant systems and addressed the question of characterizing the
importance for controllability of individual agent connections.  We
have considered a suite of performance metrics based on the spectral
properties of the associated controllability Gramian and formally
characterized the effect on them of perturbing the weights of all
possible edges, including those not present in the network. This
analysis has led us to propose a novel notion of edge centrality as a
way of measuring the first-order variation in network performance.
The edge centrality matrix (ECM) encodes important physically
realizable quantities and is additive on the set of inputs, meaning
that it captures the specific contribution to each edge's centrality
of the presence of any given actuator.  We have fully characterized
the structure of ECM for the class of directed stem-bud networks,
which possess a diagonal controllability Gramian, and shown that it
only depends on the network size and the length of its bud, with
possible non-zero entries only at specific sub/super-diagonals.
Finally, given an edge modification weight, we have developed novel
bounds on the value of the trace, trace inverse, and log-det of the
Gramian before and after single-edge modifications. We have also
determined bounds on the weight that ensure the resulting modified
network remains stable. Numerical examples illustrate the usefulness
of the proposed edge centrality notion and the derived results.
Future work will employ the proposed notions in the algorithmic
synthesis of stable networks via edge modification with enhanced
guarantees on convergence and controllability and the mitigation of
the effect of malicious attacks in network strategic scenarios,
explore the design of distributed schemes for the computation of the
proposed edge centrality measures, and develop bounds of the impact on
network performance of multiple-edge modifications.

\appendix
\label{appendix1}

\begin{lemma}\longthmtitle{Upper bound on largest eigenvalue of sum of
    two rank-one matrices}\label{lem4}
  For $v\in \real^n$, let $U=v e_i^{\top}+e_i v^{\top}$.  Then,
  \begin{align*}
    \lambda_1(U) & =2v^{\top}e_i && \text{if } \rank(U)=1,
    \\
    \lambda_1 (U) & \leq v^{\top}e_i+\Vert v\Vert && \text{if }
    \rank(U)=2.
  \end{align*}
\end{lemma}
\begin{proof}
  Denote $\rank(U)=r_U$.  From \cite[Theorem 2.1]{HW-GPHS:80},
  \begin{align}
    \label{peq-lem4-2}
    \lambda_1(U)\leq \frac{\tr(U)}{r_U} +
    \sqrt{(r_U-1)\Bigg(\frac{\tr(U^2)}{r_U}-\Bigg(\frac{\tr(U)}{r_U}\Bigg)^2\Bigg)}.
  \end{align}
  Now,
  \begin{align}
    \label{peq-lem4-3}
    \notag \tr(U) & =\tr(ve_i^{\top})+\tr(e_iv^{\top}) =
    2v^{\top}e_i,
    \\
    \tr(U^2) &=2(v^{\top}e_i)^2+2v^{\top}v.
  \end{align}
  As $\rank(ve_i^{\top})=\rank(e_iv^{\top})=1$, from \cite[Section
  0.4.5]{RAH-CRJ:12}, $\rank(U)\leq 2$. Using \eqref{peq-lem4-3} and
  $\Vert v\Vert=\sqrt{v^{\top}v}$ in \eqref{peq-lem4-2} for
  $r_U=1$ and $r_U=2$ gives the required result.
\end{proof}

\begin{lemma}\longthmtitle{Bounds on  norm of positive definite
    matrix}\label{lem5} 
  Let $\Zc$ be a $n\times n$ symmetric positive definite matrix and $
  \lo \geq \lambda_1(\Zc)$. Then,
  \begin{align*}
    \frac{(\tr(\Zc))^2}{n}\leq \Vert \Zc \Vert^2 \leq \lambda_1(\Zc)
    \tr(\Zc) \leq \lo \tr(\Zc) .
  \end{align*}
\end{lemma}
\begin{proof}
  From \cite[Proposition 8.4.13]{DSB:09}, $\Vert \Zc
  \Vert^2=\tr(\Zc^2)\leq\lambda_1(\Zc)\tr(\Zc) \leq \lo
  \tr(\Zc)$. From \cite[Fact 8.12.1]{DSB:09},
  $\frac{(\tr(\Zc))^2}{n}\leq \tr(\Zc^2)=\Vert \Zc \Vert^2$.
\end{proof}

\bibliographystyle{IEEEtran}%
\bibliography{alias,JC,Main,Main-add}

\vspace*{-8ex}

\begin{IEEEbiography}[{\includegraphics[width=1in,height=1.25in,clip,keepaspectratio]{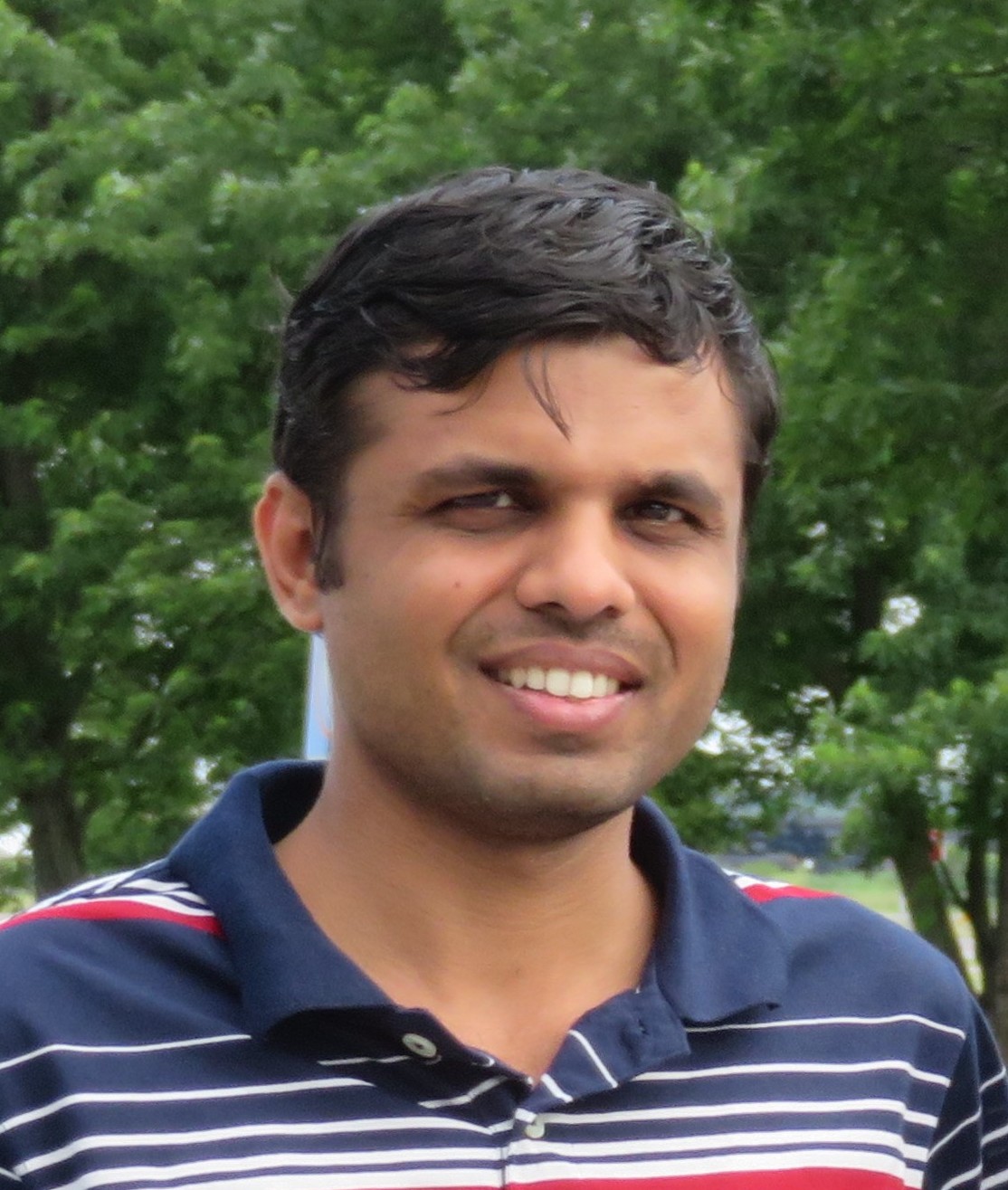}}]{Prasad
    Vilas Chanekar}
  received the undergraduate degree in Mathematics and Computing from
  the Indian Institute of Technology, Kharagpur, India. He then
  obtained M.Sc. and Ph.D. degrees in Mechanical Engineering from the
  Indian Institute of Science, Bangalore, India and the University of
  Maryland, College Park, USA respectively. He is currently a
  Post-Doctoral researcher in the Department of Mechanical and
  Aerospace Engineering at the University of California, San Diego,
  USA. His research interests include dynamics, control and
  optimization of large-scale engineering systems.
\end{IEEEbiography}

\vspace*{-6ex}

\begin{IEEEbiography}[{\includegraphics[width=1in,height=1.25in,clip,keepaspectratio]{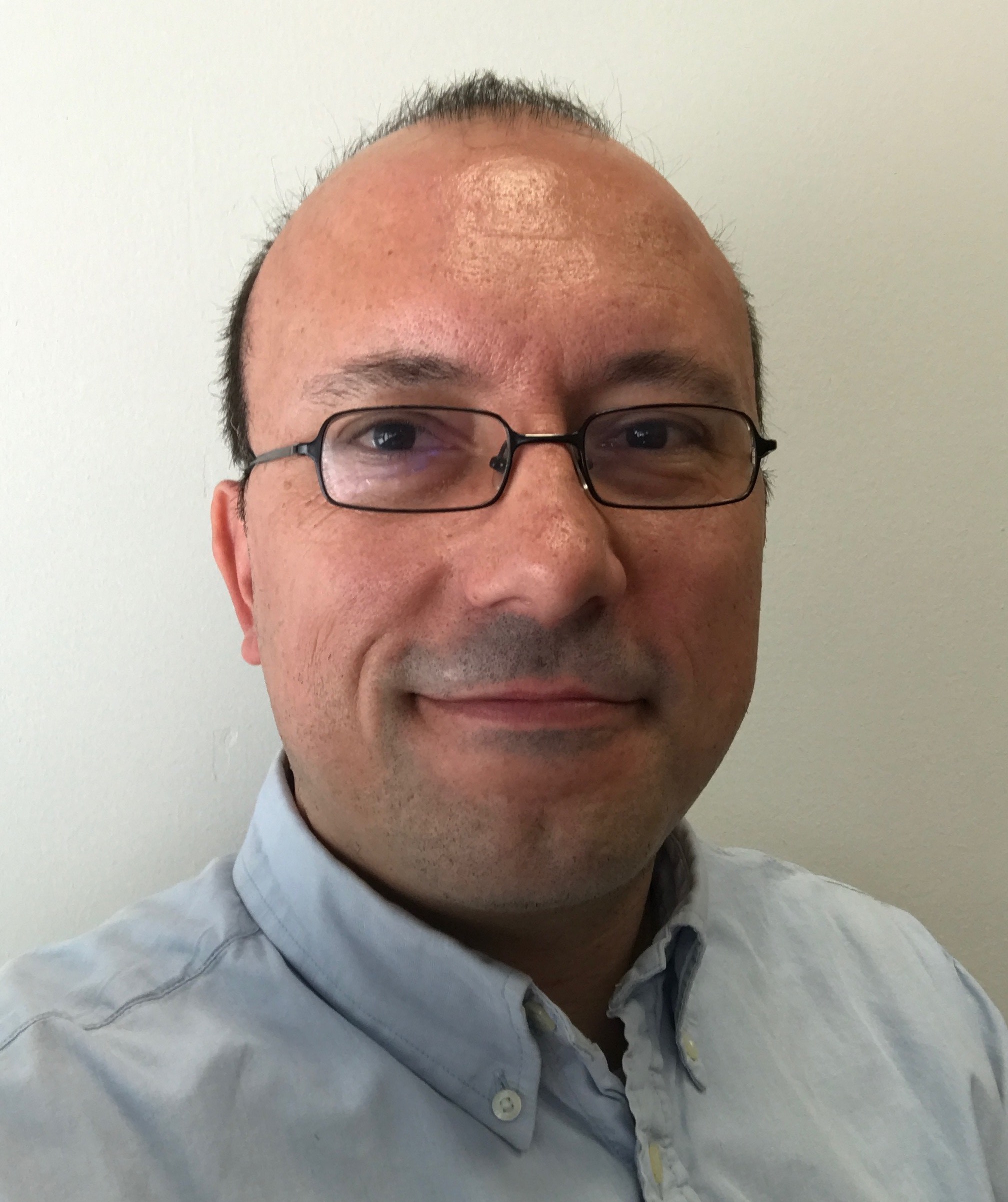}}]{Jorge
    Cort\'{e}s}
  (M'02, SM'06, F'14) received the Licenciatura degree in mathematics
  from Universidad de Zaragoza, Zaragoza, Spain, in 1997, and the
  Ph.D. degree in engineering mathematics from Universidad Carlos III
  de Madrid, Madrid, Spain, in 2001. He held postdoctoral positions
  with the University of Twente, Twente, The Netherlands, and the
  University of Illinois at Urbana-Champaign, Urbana, IL, USA. He was
  an Assistant Professor with the Department of Applied Mathematics
  and Statistics, University of California, Santa Cruz, CA, USA, from
  2004 to 2007. He is currently a Professor in the Department of
  Mechanical and Aerospace Engineering, University of California, San
  Diego, CA, USA.  He is the author of Geometric, Control and
  Numerical Aspects of Nonholonomic Systems (Springer-Verlag, 2002)
  and co-author (together with F. Bullo and S.  Mart{\'\i}nez) of
  Distributed Control of Robotic Networks (Princeton University Press,
  2009).  He is a Fellow of IEEE and SIAM.
  His current research interests include distributed control and
  optimization, network science, nonsmooth analysis, reasoning and
  decision making under uncertainty, network neuroscience, and
  multi-agent coordination in robotic, power, and transportation
  networks.
\end{IEEEbiography}

\end{document}